\documentclass[11pt]{article}

\usepackage{natbib}
\usepackage{makeidx}         % allows index generation
\usepackage{appendix}
\usepackage[normalem]{ulem}
\usepackage[colorlinks=true, allcolors=blue]{hyperref}
\usepackage{blkarray}

\usepackage[utf8]{inputenc} % Use UTF-8 encoding for the input
\usepackage[T1]{fontenc}    % Use T1 font encoding
\usepackage{lmodern}        % Use Latin Modern fonts

\usepackage{mathptmx}       % Times Roman font for text
\usepackage{helvet}         % Helvetica font for sans-serif
\usepackage{courier}        % Courier font for monospaced text
\usepackage{type1cm}        % Improve font quality

\usepackage{graphicx}        % Include images
\usepackage{subfig}          % Support for subfigures
\usepackage{tikz}            % Create vector graphics directly in LaTeX
\usepackage{tikz-cd}         % Commutative diagrams with TikZ

\usepackage{booktabs}        % Professional-quality tables

\usepackage{algorithm}
\usepackage[noend]{algpseudocode}

\usepackage{amsmath}         % American Mathematical Society (AMS) Math package
\usepackage{amsfonts}        % AMS fonts
\usepackage{amssymb}         % AMS symbols

\usepackage{amsthm}
\newtheorem*{proposition}{Proposition}

\usepackage{bbm}             % For blackboard bold math
\usepackage{bm}              % For bold math symbols
\usepackage{makeidx}         % Support for index creation
\usepackage{multicol}        % Multiple columns in the index
\usepackage[bottom]{footmisc}% Footnotes at the bottom of the page

\newcommand{\Initialization}{\State \textbf{Initialization: }}
\newcommand{\Data}{\State \textbf{Data: }}

\title{Inferring the dynamics of quasi-reaction systems via nonlinear local mean-field approximations}
\author{Matteo Framba\textsuperscript{1}, Veronica Vinciotti\textsuperscript{1}, Ernst C. Wit\textsuperscript{2}\\
	\textsuperscript{1}\textit{ University of Trento}, 
	\textsuperscript{2}\textit{ Universit\`a della Svizzera italiana}%\\
}
\date{}

\begin{document}
	\maketitle

\begin{abstract}
In the modelling of stochastic phenomena, such as quasi-reaction systems, parameter estimation of kinetic rates can be challenging, particularly when the time gap between consecutive measurements is large.  Local linear approximation approaches account for the stochasticity in the system but fail to capture the nonlinear nature of the underlying process. 
At the mean level, the dynamics of the system can be described by a system of ODEs, which have an explicit solution only for simple unitary systems. An analytical solution for generic quasi-reaction systems is proposed via a first order Taylor approximation of the hazard rate. This allows a nonlinear forward prediction of the future dynamics given the current state of the system. Predictions and corresponding observations are embedded in a nonlinear least-squares approach for parameter estimation. The performance of the algorithm is compared to existing SDE and ODE-based methods via a simulation study.
Besides the increased computational efficiency of the approach, the results show an improvement in the kinetic rate estimation, particularly for data observed at large time intervals. Additionally, the availability of an explicit solution makes the method robust to stiffness, which is often present in biological systems. An illustration on Rhesus Macaque data shows the applicability of the approach to the study of cell differentiation.
\end{abstract}

\section{Introduction}
\label{sec:1}

Reaction networks are an efficient framework used to describe the population evolution in many biological and biochemical phenomena. These systems are typically modeled using stochastic differential equations, which effectively capture the inherent uncertainty and randomness of the underlying biological structure \citep{golightly2005bayesian}. Understanding the dynamics of a process requires a thorough knowledge of the evolution of its moments, typically obtained through the chemical master equation \citep{schnakenberg1976network}. The primary objective of many studies is to infer the parameters governing these moments, often achieved using Local Linear Approximation (LLA) due to its efficiency and ease of implementation. However, LLA methods have been found to be inaccurate when data are obtained within very small time intervals, due to collinearity, or across very large time intervals, due to nonlinearity. The former problem was recently addressed by means of a state space formulation involving modelling latent reactions \citep{framba2024latent}. With respect to the latter problem, existing methods exhibit significant estimation bias because of poor approximations \citep{shoji2013nonparametric}. This is a serious problem, as large observation intervals are typical in many practical experimental settings, such as gene therapy clonal studies, where blood sampling occurs monthly so as to align with the months-long lifespan of blood cells and their production cycle \citep{pellin2023tracking}.

	To date, only few studies have explored the challenges of parameter inference in quasi-reaction models for widely-spaced-data. \cite{pellin2023tracking} and \cite{milner2013moment} proposed moment-closure methods that numerically solve the differential equations of the first and second moment of the process, but this requires considerable computational effort especially for large population sizes.  The Bayesian inference approach in \cite{boys2008bayesian} works well in data-poor scenarios, but is computationally inefficient. Mean-field approximation techniques \citep{baccelli1992mean} offer a viable alternative by providing explicit solutions for the first moments of state distributions while maintaining the process's nonlinearity. However, this approach is limited to unitary systems, where each reaction involves the transformation of a single element into one or more products. Such scenarios are rare, as real-world models are better characterized by nonlinear dynamics, which capture complex behaviors like logistic growth \citep{tsoularis2002analysis}, bifurcations \citep{hale2012dynamics}, and limit cycles \citep{ye1986theory}. \cite{xu2019statistical} proposed a method-of-moments algorithm that matches second order moments. However, due to the statistical instability of the second moment, this approach fails in highly stochastic or nonlinear systems.

	 Various methods have employed Taylor approximations to solve kinetic rate equations. \cite{kennealy1977numerical} utilized the Taylor series expansion for numerical integration of chemical kinetics, leveraging the ease of obtaining higher-order derivatives from the specific form and symmetries of differential equations in chemical systems. However, the need for frequent adjustments to the step size to ensure convergence remains a challenge. \cite{cordoba1998optimal} introduced a method for optimizing the initial parameters of the Taylor integrator by controlling local errors. This approach allows for larger step sizes without increasing computational complexity significantly, yet it demands an accurate analytical expression for the local errors, which can be difficult to obtain in complex systems. \cite{lente2022use} developed an algorithm based on Taylor's theorem for solving kinetic differential equations, using polynomial expansions of concentration-time functions. However, its applicability is limited by the requirement for suitable time transformations to maintain the polynomial nature of the rate equations, which may not always be feasible. In this paper,  we propose an efficient method that extends the mean-field approach using a Taylor expansion not in time, as the above methods proposed, but in concentration.  Starting from the chemical master equation and using the system of nonlinear equations describing the dynamics of the process mean, we obtain a linear approximation of the rate function. This leads to an approximation of the system of ordinary differential equations (ODEs) with an explicit solution. By combining our method with a nonlinear least-squares method, it is possible to perform inference of the parameters governing the rate equations.  

The paper is structured as follows. In section \ref{sec:2}, we formalize the statistical modelling of quasi-reaction systems and introduce the generic mean-field approach. We define the proposed local mean-field approximation method and illustrate it in an example. We then study its resistance to stiffness, by comparing the performance of our method against several numerical approaches. The nonlinear least-squares procedure for parameter estimation is described in section~\ref{sec:inference}, both from a methodological and computational point of view. Section~\ref{sec:5} is reserved for simulation studies. By comparing the performance of the proposed algorithm with both the existing LLA method and other state-of-the-art approaches, we show its improved performance, particularly as the time interval between consecutive observations increases. In section~\ref{sec:rhesus}, we illustrate the method in the analysis of a cell differentiation study by \cite{wu2014clonal}. Via BIC model selection, we study the pattern of cell differentiation of 5 major blood cells, and estimate the parameters regulating the underlying reaction process.

\section{Local Mean-field Approximation (LMA)} 
\label{sec:2}

In this section, we describe the proposed method for estimation of parameters in a quasi-reaction system. To this end, in section~\ref{sec:QRM},  we first describe concisely a general quasi-reaction system, with the aim of deriving the general ODE formulation of the conditional mean of this process. In section~\ref{sec:unitary}, we show that this system of ODEs can be solved explicitly for unitary systems. For generic systems, the solution does not exist. However, by using the solution from a unitary system, we derive in section~\ref{sec:LMA} a generic approximation for any quasi-reaction system.  

\subsection{Quasi-reaction models} \label{sec:QRM}
Consider a closed system with $p$ interacting species $\{Q_1,\dots,Q_p \}$.  Every interaction between the substrates is caused by the occurrence of a quasi-reaction  $R_j$, described as 
\begin{equation}
	\label{initial}
	k_{1j}Q_1+... +k_{pj}Q_p \xrightarrow{\theta_j}  s_{1j}Q_1+... +s_{pj}Q_p \quad\quad j\in1,\dots,r.
\end{equation}
The occurrence of the $j$-th reaction leads to a change of ${v}_{lj}={s}_{lj}-{k}_{lj}$ substrates for particle type $l$. Let $V$ denote the net effect matrix having $v_{lj}$ as $lj$-th element, and $K = \{k_{lj}\}$ the reactant matrix.
Let $\bm{Y}(t)=(Y_{1}(t),\dots,Y_{p}(t))^T\in\mathbb{N}_0^p$ denote the continuous time counting process, with $Y_l(t)$ the number of $l$-th particles  present in the system at time $t\in[0,T]$. Under standard assumptions, the conditional rate for the $j$-th reaction in such system is given as
\begin{equation}
	\label{hazard}
	\lambda_j(\bm{Y}(t);\boldsymbol{\theta}) = \theta_j \prod_{l=1}^p  \binom{Y_l(t)}{k_{lj}},
\end{equation}
with $\binom{Y_{l}(t)} {k_{lj}}=0, $ for all $\hspace{0.05cm}{Y_{l}(t) <k_{lj}}$. The rate parameters $\boldsymbol{\theta} \in \mathbb{R}^{r}_+$ govern the process dynamics.  Let \( \Theta \) be a diagonal matrix with \(\bm{\theta} \) on the diagonal, and \( \bm{\kappa}(t) \)  the vector with \( j \)-th element \( \kappa_j(t)= \prod_{l=1}^{p} \binom{Y_l(t)}{k_{lj}} \). Thus, the hazard function can be compactly written as \( \bm{\lambda}(\bm{Y}(t); \bm{\theta}) = \Theta \bm{\kappa}(t) \).

The stochastic process  $\bm{Y}(t)$ obeys the Markov property. Given an initial condition $\bm{y}(t_0)$, it is possible to determine the probability density $p_t(\bm{Y})$ of the system being in the state $\bm{Y}$ at time $t$. The temporal evolution of the Markov process transition kernel is governed by the Kolmogorov's forward equations \citep{wilkinson18}. In the context of stochastic kinetic process, these  are commonly referred to as the chemical master equation, which is given by
\begin{equation}
	\label{CME}
	\frac{d p_t(\bm{y})}{dt}=\sum_{j=1}^{r}p_t(\bm{y}-\bm{v}_{\cdot j})\lambda_j(\bm{y}-\bm{v}_{\cdot j };\theta_j)-p_t(\bm{y})\lambda_j(\bm{y};\theta_j), \quad \forall\hspace{0.1cm} \bm{y}\in \mathbb{N}_0^p.
\end{equation}
Solving the above equation involves assessing the evolution of $P_t(\bm{y})$  over the entire range of possible configurations for the process. This approach clearly does not offer a feasible solution for systems of realistic size and complexity. Nevertheless, valuable insights regarding the dynamics of characteristic statistical features of the system can be derived from \eqref{CME}. 

In particular, one can obtain from the chemical master equation a set of ODEs describing the temporal evolution of lineage population concentration averages.  To this end, let $\bm{m}(t+s|t)$ describe the evolution of $\mathbb{E}[\bm{Y}(t+s)|\bm{Y}(t)=\bm{y}(t)]=\sum_{\bm{y}}\bm{y}p_{t+s|t}(\bm{y})$. By taking the derivative, the following ODEs system describing the dynamics of the conditional mean of $\bm{Y}(t+s)|\bm{Y}(t)$ can be obtained \citep{pellin2023tracking}:
 \begin{equation}
 	\label{mean1}
 	\dfrac{d m_{l}(t+s|t)}{ds}= \sum_{j=1}^r v_{lj}\mathbb{E}[\lambda_j(\bm{Y}(t+s);\boldsymbol{\theta})~|~\bm{Y}(t)], \quad l=1,\ldots,p.
  \end{equation} 
%A detailed derivation of \eqref{mean1} can be found in \cite{pellin2023tracking}. 
There have been several proposals for solving \eqref{mean1} using approximation methods. These, however, come with significant limitations. The system size expansion method  by \cite{van1992stochastic} tends to lose accuracy in systems with small populations or complex dynamics. The moment closure approximation \citep{grima2012study} often loses information due to the arbitrary truncation of higher moments, while the diffusion approximation method by \cite{golightly2005bayesian} can be unreliable in cases with low molecule counts and highly nonlinear behaviour. In the next section, we show how an explicit mean-field solution  can be found for unitary systems, and how this can then be used as the basis of a new approximation method.

\subsection{Explicit mean-field solution for unitary systems}
\label{sec:unitary}

Unitary systems are quasi-reaction systems where each reaction needs at most one particle from each reactant in order to occur. In such systems, the hazard rate \eqref{hazard} is linear in $\bm{y}$, so \[\mathbb{E}[\lambda_j(\bm{Y}(t+s;\boldsymbol{\theta}))~|~\bm{Y}(t)=\bm{y}(t)]= \lambda_j(\mathbb{E}[\bm{Y}(t+s)~|~\bm{Y}(t)=\bm{y}(t)];\boldsymbol{\theta})=\lambda_j(\bm{m}(t+s|t);\boldsymbol{\theta}).\] By adding the initial condition $\bm{m}(t|t)=\bm{y}(t)$, the system (\ref{mean1}) is expressed by the following first order Cauchy differential equation, 
\begin{equation}
	\left\{
	\begin{aligned}
		\label{cauchy}
		\frac{d\bm{m}(t+s|t)}{ds} &= P_{\boldsymbol{\theta}}\bm{m}(t+s|t) + \bm{b}_{\boldsymbol{\theta}} \\
		\bm{m}(t|t) &= \bm{y}(t)
	\end{aligned}
	\right.
\end{equation}
The coefficient matrix $P_{\boldsymbol{\theta}}$ and the inhomogeneous term $\bm{b}_{\boldsymbol{\theta}}$ are functions of the vector $\boldsymbol{\theta}$. The former relates to the reactions that involve exactly one reactant, whereas the latter refers to spontaneous reactions that do not involve any reactants. By employing the hazard function \eqref{hazard} and applying straightforward algebraic manipulations, they can be explicitly expressed in terms of the death and net effect matrices as  $P_{\boldsymbol{\theta}}=V\Theta K^T$ and
	$b_{\boldsymbol{\theta},l}\notag=\sum_j V_{l j}\theta_j\mathbbm{1}_{\{K_{.j}=0 \}}$.
If the $P_{\boldsymbol{\theta}}$ matrix is invertible, the system \eqref{cauchy} has the explicit solution 
\begin{equation}
	\label{ODEsol}
\bm{m}(t+s|t)=\exp\big(sP_{\boldsymbol{\theta}}\big)\bm{y}(t) + {P_{\boldsymbol{\theta}}}^{-1}\bigg(\exp\big(sP_{\boldsymbol{\theta}}\big)-\mathbb{I}_p\bigg)\bm{b}_{\boldsymbol{\theta}}.
\end{equation}
Although this approach only works for unitary systems, it inspires our proposal for generic quasi-reaction systems explained in the next section.

\subsection{LMA: local mean-field approximation for generic systems}\label{sec:LMA}
Most biological quasi-reaction systems involve complex, higher order interactions that make finding an analytical solution to \eqref{mean1} typically impossible. Our idea is to linearise the hazard function with respect to the abundance vector $\bm{Y}$ so that any quasi-reaction system can be approximated as a unitary system, thus obtaining an explicit solution of the ODEs system of the form \eqref{ODEsol}. Omitting the dependence of the hazard function on the parameters $\boldsymbol{\theta}$ for the sake of readability, we perform a first order Taylor expansion of \( \boldsymbol{\lambda}(\bm{Y}(t+s)) \) around \( \bm{Y}(t) \), 
\begin{equation}
\label{Taylorapprox}
	\boldsymbol{\lambda}(\bm{Y}(t+s))= \boldsymbol{\lambda}(\bm{Y}(t))+\Lambda (\bm{Y}(t+s)-\bm{Y}(t))+\boldsymbol{\eta}(t).
\end{equation}
Here, \( \eta_j \) is the approximation error for the \( j \)-th component, which, from Taylor's theorem, is of the form
\begin{equation}
	\label{err}
	\eta_j = \frac{1}{2} \frac{\partial^2 \lambda_j( \bm{\tilde{Y}})}{\partial Y_l(t) \partial Y_k(t)} (Y_l(t+s) - Y_l(t))(Y_k(t+s) - Y_k(t)), 
	\quad l,k = 1, \dots, p
\end{equation}
with \( \bm{\tilde{Y}} \in (\bm{Y}(t), \bm{Y}(t+s)) \). The Jacobian matrix $\Lambda=\dfrac{\partial \bm{\lambda}}{\partial \bm{y}}$ is explicitly defined by the  proposition below, whose proof can be found in Appendix \ref{jac:der}.
\begin{proposition}
	Given the intensity function $
	\lambda_j(\bm{Y}(t);\boldsymbol{\theta}) = \theta_j \prod_{l=1}^p \binom{Y_l(t)}{k_{lj}},
	$ the $jl$-th element of the Jacobian matrix $\Lambda(\mathbf{Y}(t);\boldsymbol{\theta})\in \mathbb{R}^{r\times p}$ is given by
	\begin{equation*}
		\Lambda_{jl}=\theta_j\underbrace{\prod_{i=1}^{p}\binom{Y_i(t)}{k_{ij}}(1-\delta_{il})\binom{Y_l(t)}{k_{ij}}\bigg(\psi(Y_l(t)+1)-\psi(Y_l(t)-k_{lj}+1)\bigg)}_{H_{jl}}
	\end{equation*}
	where
	$\psi(x)=\dfrac{d}{dx}\log(\Gamma(x)) $ is the digamma function, i.e., the logarithmic derivative of the gamma function. 
\end{proposition}

Using now approximation \eqref{Taylorapprox} in (\ref{mean1}), we obtain the conditional mean ODEs system:
\begin{align}
	\frac{d}{ds}	\bm{m}(t+s|t)&=V\mathbb{E}\bigg[\boldsymbol{\lambda}(\bm{Y}(t))+\Lambda\big(\bm{Y}(t+s)-\bm{Y}(t)\big)~|~ \bm{Y}(t)=\bm{y}(t)\bigg]\notag\\
	&=V\Lambda\mathbb{E}[\boldsymbol{\lambda}(\bm{Y}(t+s))|\bm{Y}(t)=\bm{y}(t)]+V\boldsymbol{\lambda}(\bm{y}(t))-V\Lambda\bm{y}(t)
	\notag \\
	&=\underbrace{V \,\Theta\, H}_{P_{\boldsymbol{\theta}}}\bm{m}(t+s|t)+\underbrace{V \, \Theta\, (\boldsymbol{\kappa}(t) -  H \, \bm{y}(t))}_{\bm{b}_{\boldsymbol{\theta}}}
	\notag \\
	&=P_{\boldsymbol{\theta}}\bm{m}(t+s|t)+\bm{b}_{\boldsymbol{\theta}}. \label{ODEmean} 
\end{align}  
If $|P_{\boldsymbol{\theta}}|\neq 0$, the explicit solution of the ODEs system above is defined as in (\ref{ODEsol}).

In terms of convergence to the solution of the original ODEs system \eqref{mean1}, the approximation does not depend on the time step but on the difference between concentrations as indicated by the relation \eqref{err}. Since we stop the expansion at first order, and if we consider at most second-order reactions, then the approximation error is determined by the norm of the Hessian matrix of the hazard function.

\subsection{Example: cyclic chemical reaction network}\label{example}
In this section, we consider the example of a cyclic chemical reaction network involving three particle types $(A,B,C)$ with abundances $\bm{Y}=(Y_1,Y_2,Y_3)$. The reactions are given by
	\begin{align*}
		2 \text{A} & \xrightarrow{\theta_1} 2\text{B} \\
		\text{A} + \text{B} & \xrightarrow{\theta_2} 3\text{C}  \\
		2 \text{C} & \xrightarrow{\theta_3} 2\text{A}.
	\end{align*}
In particular, this network is a closed loop where the products of one reaction act as the reactants for another, creating a continuous cycle of chemical transformations. Such cyclic networks are central in understanding metabolic cycles and oscillatory behavior in biological systems  \citep{gillespie2007stochastic}. The first reaction describes the conversion of two molecules of species A into two molecules of species B and could be seen as a simple process where A is transformed into B in the presence of a catalyst. The second reaction, where A and B combine to form three molecules of species C, can be viewed as a synthesis reaction often seen in polymerization processes. The final reaction regenerates species A from C, completing the cycle and ensuring the continuity of the process. 

The reactant and net effect matrix associated to this system are given, respectively, by
\[
K=\begin{bmatrix}2&1&0\\0&1&0\\0&0&2\end{bmatrix}, \quad V=\begin{bmatrix}-2&-1&2\\2&-1&0\\0&3&-2\end{bmatrix}.
\]
The first is used in the definition of the hazard rates in \eqref{hazard}, while the second, together with the hazard function, define the  dynamics of the first moments of the concentrations in \eqref{mean1}.
%which, when used in (\ref{mean1}), leads to the following ODEs system for the dynamics of the first moments of the concentrations	
%\begin{align}
		%\label{CCRN1}
		%\frac{d{m}_1(t)}{dt} &= -2\theta_1 m_1(t)^2 + \theta_2 m_1(t) m_2(t) + 2 \theta_3 m_3(t)^2 \notag\\
		%\frac{d{m}_2(t)}{dt} &= 2\theta_1 m_1(t)^2 - \theta_2 m_1(t) m_2(t)  \\
		%\frac{d{m}_3(t)}{dt} &= 3\theta_2 m_1(t) m_2(t) - 2 \theta_3 m_3(t)^2.  \notag
	%\end{align}

The quasi-reaction system is clearly non-unitary, as more than one particle type is used in each reaction. This leads to a hazard that is not linear in $\bm{Y}$ and to no analytical solution to the ODEs system in \eqref{mean1}. The local mean-field approximation method described in section~\ref{sec:LMA} provides an explicit approximate solution to the system. In order to show this, let \( \bm{y}=(y_1,y_2,y_3) \) represent the observation of the continuous process at time \( t \).
The reaction events are associated with the hazard function 
\[ \bm{\lambda}(\bm{y}) := \lambda(\bm{Y}(t);\boldsymbol{\theta})\big|_{\bm{Y}(t)=\bm{y}} =\begin{bmatrix}
	\theta_1 y_1(y_1 - 1)/2 \\
	\theta_2 y_1 y_2 \\
	\theta_3 y_3(y_3 - 1)/2
\end{bmatrix},  \]
for which we consider a Taylor expansion in $\bm{y}$ and then its first order approximation. In particular, the Jacobian matrix, evaluated at $\bm{y}$, is given by 
\[
\Lambda = \frac{\partial \bm{\lambda}(\bm{Y(t)};\boldsymbol{\theta})}{\partial \bm{Y}}\big|_{\bm{Y}(t)=\bm{y}} =\begin{bmatrix}\theta_1(y_1-0.5)&0&0\\\theta_2y_2&\theta_2y_1&0\\0&0&\theta_3(y_3-0.5)\end{bmatrix}.
\]

We now plug in these quantities in \eqref{Taylorapprox}, leading to the ODEs system approximation \eqref{ODEmean} with
\begin{align*}
	P_{\boldsymbol{\theta}}&=V\Lambda =\begin{bmatrix}-2&-1&2\\2&-1&0\\0&3&-2\end{bmatrix}\begin{bmatrix}\theta_1(y_1-0.5)&0&0\\\theta_2y_2&\theta_2y_1&0\\0&0&\theta_3(y_3-0.5)\end{bmatrix}\\
	&=\begin{bmatrix}-2\theta_1(y_1-0.5)-\theta_2y_2&-\theta_2y_1&2\theta_3(y_3-0.5)\\2\theta_1(y_1-0.5)-\theta_2y_2&-\theta_2y_1&0\\3\theta_2y_2&3\theta_2y_1&-2\theta_3(y_3-0.5)\end{bmatrix},\\
	\bm{b}_{\boldsymbol{\theta}}&=V(\lambda(\bm{y})-\Lambda\bm{y})\\
	&=\begin{bmatrix}-2&-1&2\\2&-1&0\\0&3&-2\end{bmatrix}\bigg(\begin{bmatrix}
		\theta_1 y_1(y_1 - 1)/2 \\
		\theta_2 y_1 y_2 \\
		\theta_3 y_3(y_3 - 1)/2
	\end{bmatrix} -\begin{bmatrix}\theta_1(y_1-0.5)&0&0\\\theta_2y_2&\theta_2y_1&0\\0&0&\theta_3(y_3-0.5)\end{bmatrix}\begin{bmatrix}
		y_1\\y_2\\y_3
	\end{bmatrix}\bigg)\\&=\begin{bmatrix}
		\theta_1(y_1)^2+\theta_2y_1y_2-\theta_3(y_3)^2\\-\theta_1(y_1)^2+\theta_2y_1y_2\\-3\theta_2y_1y_2+\theta_3(y_3)^2
	\end{bmatrix}.
\end{align*}

By substituting the quantities above into equation \eqref{ODEsol}, we obtain the explicit form of the mean process values after a time interval $s$. These constitute a nonlinear forward prediction of the system at time $t+s$. It should be noted that \( \det(P_{\boldsymbol{\theta}}) = 12\theta_1\theta_2\theta_3 (y_1 - 0.5) y_2 (y_3 - 0.5) \) is non-zero if and only if \( \boldsymbol{\theta} \neq 0 \) and \( \bm{y} \neq (0.5,0,0.5) \).

\subsection{Computational time and stiffness}\label{stiffness}
The main advantage of the proposed method is that the ODEs system in \eqref{ODEmean}, used to approximate the orginal system in \eqref{mean1}, has an explicit solution. As we will show in this section, classical numerical methods for solving ODEs systems may be computationally more efficient than calculating this analytical solution, but they are less stable in many scenarios.  

A traditional numerical algorithm for solving ODEs systems, well known for its simplicity of implementation, is the explicit Euler method. This has a slow linear convergence with respect to the width of the time steps $\Delta t$ but a linear computational cost with respect to the total number of subintervals ${T_c}=T/\Delta t$, with $t \in [0,T]$. A good alternative is the explicit fourth-order Runge-Kutta method. This offers a fourth-order convergence, by evaluating the function at multiple points within each time step, but consequently leading to a four times higher computational cost \citep{krijnen2022computational}. The balance between accuracy and computational cost makes Runge-Kutta methods preferable for many practical applications where accuracy is a priority \citep{butcher2016numerical}. As for the proposed LMA, in order to solve the ordinary differential equations \eqref{ODEmean} analytically, it is necessary to compute the exponential of a \( p \times p \) matrix.  For this, we utilize the Pade approximation with scaling and squaring method due to its efficiency with dense matrices. This approach has a computational complexity of \( \mathcal{O}(p^3) \).  Additionally, we need to compute the inversion of the same matrix, which incurs a similar computational cost. Consequently, the overall cost for calculating the analytical solution in \eqref{ODEmean} using the LMA approach amounts to \( \mathcal{O}(p^3) \), which is higher than both the Euler and Runge-Kutta methods. 

Besides computational time, a further comparison between the methods can be made in terms of performance in the presence of stiffness. It is well known how the explicit Euler and Runge-Kutta methods struggle with problems that are classified as stiff. On the other hand, the availability of an explicit solution makes the method robust also in the presence of stiffness.  A system is considered stiff in a given interval if, when applying a numerical method with a finite region of absolute stability, the step length required is excessively small relative to the smoothness of the exact solution \citep{lambert1974two}. This definition emphasizes the difficulty of maintaining stability with explicit methods, which may necessitate impractically small step sizes to accurately obtain the solution.

In quasi-reaction chemical models, the phenomenon of stiffness arises in situations where very slow and very fast reactions coexist. A well known example of a stiff problem is the kinetics of an autocatalytic system \citep{robertson1966solution}. 
However, as its associated net effect matrix %$V=(1,-1,0;0,-1,1;1,-1,0)$
is not full rank,  \( P_{\boldsymbol{\theta}} = V \Lambda \) is not invertible. Instead, we consider the cyclic chemical reaction system described in section~\ref{example}, with initial value $\bm{y}=(10,20,10)$ and reaction rate $\boldsymbol{\theta}=(2\cdot10^{-6}, 10^{-7}, 2\cdot10^{-1})$. After a first order Taylor approximation, we consider the ODEs system \eqref{ODEmean} associated to this network. The matrix associated to this system is invertible. Moreover, the eigenvalues of $P_{\boldsymbol{\theta}}$, given by
$7.6\cdot 10^{-7}$, $-4.2\cdot 10^{-5}$, and  $-3.8\cdot 10^0$, are different from each other in magnitude, which is taken as an indication of stiffness \citep{butcher2016numerical}. Intuitively, it is clear how the third reaction is much faster than the first two. %Since we do not have an explicit solution of the system, we compare the explicit Euler, Runge-Kutta and LMA methods with the numerical reference obtained by Radau's method, which is known to be robust to stiff problems but computationally expensive  \citep{hairer2015radau}.
 
Figure~\ref{fig:2} shows the robustness of the Euler and Runge-Kutta methods when used to solve numerically the ODEs system \eqref{ODEmean} of this problem.
\begin{figure}
	\centering
	\includegraphics[scale=.3]{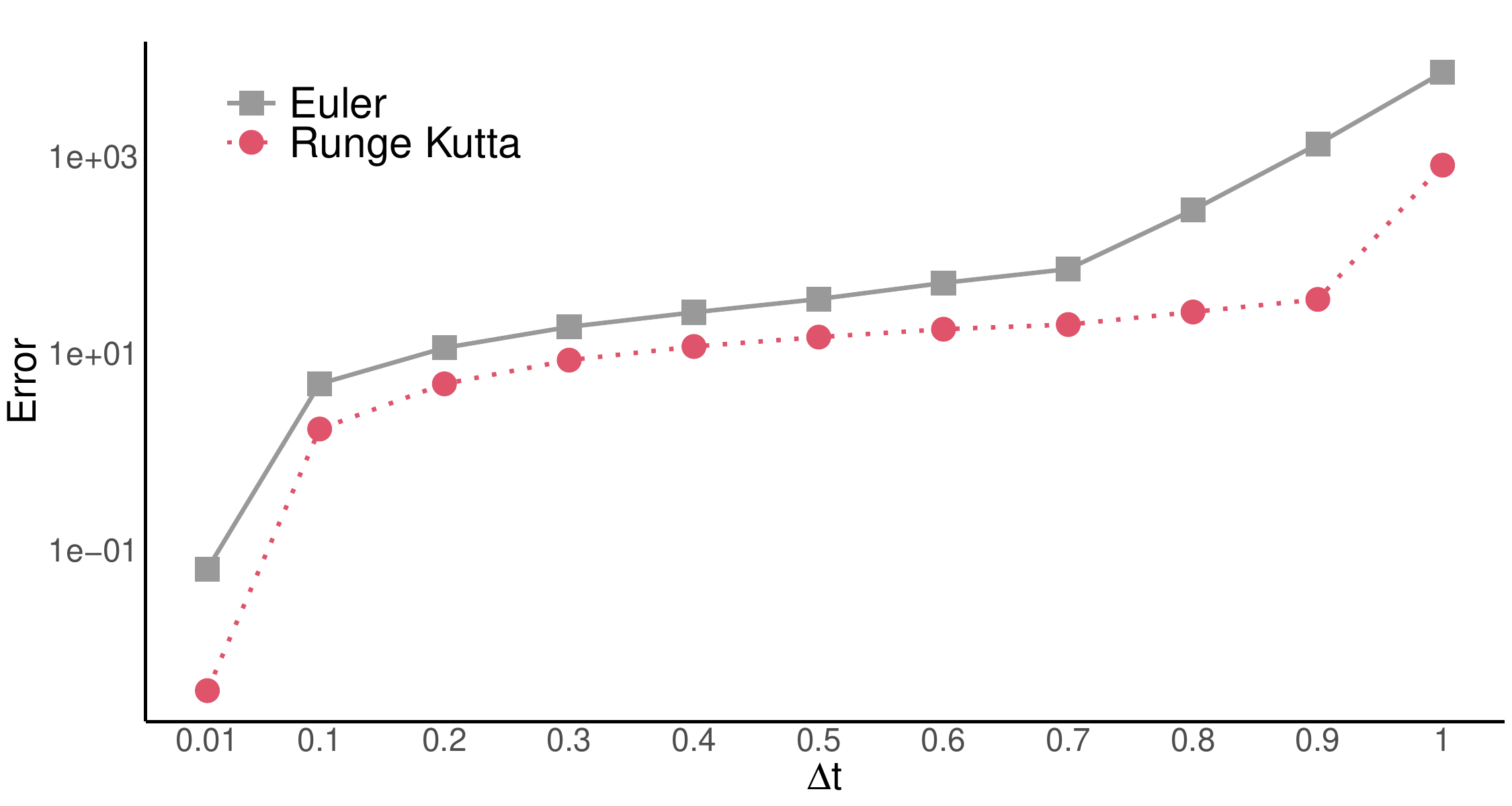}
	\caption{Mean absolute error between the numerical solution from each method and the analytical LMA solution in \eqref{ODEsol}, as $\Delta t$ increases: both Euler and Runge-Kutta methods experience a significant degradation in accuracy as $\Delta t$ increases.}
	\label{fig:2}
\end{figure}
The error is defined as the average of the absolute difference of the solution obtained by the proposed method and the analytical LMA solution in \eqref{ODEsol}, when both solutions are evaluated at $5$ time points.  For very small time intervals, the Euler and Runge Kutta methods are accurate. However, they become unstable as $\Delta t$ increases.  
Table~\ref{tab:summary} summarises the comparison of the various methods discussed in this section, in terms of computational costs, convergence and robustness to stiffness.
\begin{table}[tb]
	\centering
	\begin{tabular}{@{}lcccc@{}}
		\toprule
		\textbf{Method}                  & \textbf{Computational Cost}              & \textbf{Convergence}           & \textbf{Stiffness} \\ \midrule
		Explicit Euler                 & $\mathcal{O}(pT$)                   & $\mathcal{O}(\Delta t )$                          & Unsuitable      \\
		Runge-Kutta  & $\mathcal{O}( 4pT$)          & $\mathcal{O}(\Delta t^4)$                        & Unsuitable                            \\ 
		LMA & $\mathcal{O}( p^3 T )$           & $\mathcal{O}(1)$                        &Robust                             \\ \bottomrule
	\end{tabular}
	\caption{Comparison of computational costs, error convergence, and performance in stiff problems for the explicit Euler, fourth-order Runge-Kutta, and the proposed LMA method.}
	\label{tab:summary}
\end{table}
The proposed LMA approach provides a good compromise between computational efficiency and robustness to stiffness.

\section{Inference}\label{sec:inference}
Dynamic processes are often observed at discrete time points, and possibly across several replicates. For example, in gene therapy studies \citep{del2023scalable}, and in hematopoietic clonal dynamics \citep{pellin2023tracking}, clones, i.e., genetically identical cells, are probed at times that are days or even months apart.

\subsection{Estimation of reaction rates}
We consider a set of $n$ replicates, whereby each replicate $c$ is observed across $T_c$ time points. Let $\bm{Y}=\{ \bm{Y}_{ci}=\bm{Y}_{c}(t_{ci})\}_{c,i}^{n,T_c}$ be the set of $p$-dimensional observations of realisations of particle counts subject to the quasi-reaction system.  The time intervals are not necessary equal, which means that the observation times $t_{ci}$ are also indexed by the replicate information. 
We define $\bm{m}(\bm{\theta})=[\bm{m}_1(\bm{\theta}),\dots,\bm{m}_{n}(\bm{\theta})]$ such that  $\bm{m}_{ci}(\bm{\theta})=\mathbb{E}[\bm{Y}_{ci}~|~\bm{Y}_{c}(t_{c,i-1})=\bm{y}_c(t_{c,i-1})]$ is the solution of the ODEs system \eqref{ODEmean} with initial condition $\bm{Y}_{c,i-1}=\bm{y}_{c,i-1}$. In other words, the solution of the system of ODEs projects each observation $\bm{y}_{c,i-1}$ to the expected value at the next time point. The inference procedure for the rates $\boldsymbol{\theta}$ can then be reformulated as the following nonlinear regression problem,  
\begin{equation}
	\label{generic}
	\bm{Y}_{ci}= \bm{m}_{ci}+ \boldsymbol{\varepsilon}_{ci},
\end{equation}
where $ \boldsymbol{\varepsilon}_{ci}$ is a \textit{p}-dimensional vector of residuals such that $ \mathbb{E}[\boldsymbol{\varepsilon}_{ci}]=\textbf{0}$.  %and $\mathbb{V}[\boldsymbol{\varepsilon}_{ci}]={W}_{ci}$. The variance is itself a function of the rates $\bm{\theta}$. 
As likelihood-based approaches are unfeasible due to the computational effort involved in integrating over all possible states between the observation times, we propose instead a least-squares algorithm to estimate the $\boldsymbol{\theta}$ parameters. The constrained solution is then given by
\begin{equation}
	\label{prob}
	\hat{\boldsymbol{\theta}}_{LMA}= \arg \min_{\boldsymbol{\theta}\geq 0}\bigg\{ f(\boldsymbol{\theta})= \sum_{c=1}^{n}\sum_{i=1}^{T_c}\big[\bm{Y}_{ci}-\bm{m}_{ci}(\bm{\theta})\big]^T\big[\bm{Y}_{ci}-\bm{m}_{ci}(\bm{\theta})\big]\bigg\}.
\end{equation}

To determine the optimum, we use an iterative approach which consists of two steps. First, we solve analytically the approximation  of the ODE solution \eqref{ODEmean}. Then, we apply the limited-memory Broyden-Fletcher-Goldfarb-Shanno algorithm with box-constraints to find the optimum. This is  an optimization scheme used to solve large-scale optimization problems with simple bounds on the variables  \citep{byrd1995limited}.  A pseudo-code of the algorithm can be found in the Appendix \ref{SM:algorithm}.
The iterative procedure requires initial estimates $\hat{\boldsymbol{\theta}}_0$. Considering the potentially large number of parameters in the model, it is important to start the minimization of \eqref{prob} with accurate initial values. A practical starting value is provided by the local linear approximation approach, which will also be used in the comparative study. A detailed description of the method can be found in the Appendix~\ref{app:LLA}.

\subsection{Standard error approximation}\label{sec:sd}
In the context of statistical modeling of quasi-reaction systems, it is important to be able to evaluate the uncertainty associated with the estimated rates $\hat{\bm{\theta}}$. Only few studies provide explicit approximate formulations for this \citep{framba2024latent,tsuge2001rate}. In this section, we do so for the proposed method. Under certain regularity conditions the variance-covariance matrix of $\hat{\boldsymbol{{\theta}}}$ is approximately the inverse of the observed Fisher information matrix evaluated at $\hat{\boldsymbol{\theta}}$. The latter is the negative of the Hessian matrix $\mathbb{H}(\hat{\boldsymbol{\theta}}) := \frac{\partial^2 f(\boldsymbol{\theta})}{\partial \boldsymbol{\theta}^2}\big|_{\boldsymbol{\theta}=\hat{\boldsymbol{\theta}}}$ and can be approximated by
\begin{equation}
{\displaystyle {\mathcal {I}}(\hat{\boldsymbol{\theta}} )= \left(\frac {\partial }{\partial \theta} f(\hat{\bm{\theta}} )\right)\left({\frac {\partial }{\partial \theta}} f(\hat{\bm{\theta}} )\right)^T}.
\label{eq:sdapprox}
\end{equation}

The derivative of the objective function $f(\boldsymbol{\theta})$ with respect to the rates involves taking the derivative of the solution of the ODE system \eqref{ODEmean}, i.e., the derivatives of
\begin{equation}
	\label{mt}
	\bm{m}(t+s|t)=\exp\big(sP_{\boldsymbol{\theta}}\big)\bm{y}(t) + {P_{\boldsymbol{\theta}}}^{-1}\bigg(\exp\big(sP_{\boldsymbol{\theta}}\big)-\mathbb{I}_p\bigg)\bm{b}_{\boldsymbol{\theta}}
\end{equation}
with respect to $\boldsymbol{\theta}$.
For ease of notation, we consider a single-replicate scenario. We start by taking the derivatives of $P_{\boldsymbol{\theta}} $ and $\bm{b}_{\boldsymbol{\theta}}$  defined  in \eqref{ODEmean}. This gives
\begin{equation*}
	\label{dermiddle}
	\frac{\partial P_{\boldsymbol{\theta}}}{\partial \theta_j} = V \bm{e}_j H, \hspace{2cm}	\frac{\partial \bm{b}_{\boldsymbol{\theta}}}{\partial \theta_j} = V \bm{e}_j \boldsymbol{\kappa}(t) - V \bm{e}_j H \bm{y}(t),
\end{equation*}
where  $\bm{e}_j$ is a $r$-dimensional vector with a $1$ in the $j$-th position and $0$ elsewhere.
Using the matrix exponential derivative property and the chain rule, 
\begin{equation*}
%	\label{derexp}
	\frac{\partial \exp(s P_{\boldsymbol{\theta}})}{\partial \theta_j} = \int_0^1 \exp((1-u) s P_{\boldsymbol{\theta}}) \frac{\partial (sP_{\boldsymbol{\theta}})}{\partial \theta_j} \exp(u s P_{\boldsymbol{\theta}}) \, du.
\end{equation*}
Using the chain rule, we also obtain the derivative of the second term in (\ref{mt}),
\begin{align*}
	%\label{derinv}
	\frac{\partial}{\partial \theta_j} \left( P_{\boldsymbol{\theta}}^{-1} \left( \exp(s P_{\boldsymbol{\theta}}) - \mathbb{I}_p \right) \bm{b}_{\boldsymbol{\theta}} \right) &= \frac{\partial P_{\boldsymbol{\theta}}^{-1}}{\partial \theta_j} \left( \exp(s P_{\boldsymbol{\theta}}) - \mathbb{I}_p \right) \bm{b}_{\boldsymbol{\theta}} \nonumber + P_{\boldsymbol{\theta}}^{-1} \frac{\partial}{\partial \theta_j} \left( \exp(s P_{\boldsymbol{\theta}}) - \mathbb{I}_p \right) \bm{b}_{\boldsymbol{\theta}} \nonumber \\
	&\quad + P_{\boldsymbol{\theta}}^{-1} \left( \exp(s P_{\boldsymbol{\theta}}) - \mathbb{I}_p \right) \frac{\partial \bm{b}_{\boldsymbol{\theta}}}{\partial \theta_j}.
\end{align*}
Combining everything together and using 
$
\dfrac{\partial P_{\boldsymbol{\theta}}^{-1}}{\partial \theta_j} = -P_{\boldsymbol{\theta}}^{-1} \dfrac{\partial P_{\boldsymbol{\theta}}}{\partial \theta_j} P_{\boldsymbol{\theta}}^{-1}
$, the partial derivative of the conditional predicted values is
\begin{align}
	\frac{\partial \bm{m}(t+s \mid t)}{\partial \theta_j} &= \left(s \int_0^1 \exp((1-u) s P_{\boldsymbol{\theta}}) V \bm{e}_j H \exp(u s P_{\boldsymbol{\theta}}) \, du \right) \bm{y}(t) - P_{\boldsymbol{\theta}}^{-1} V {e}_j H P_{\boldsymbol{\theta}}^{-1} \left( \exp(s P_{\boldsymbol{\theta}}) - \mathbb{I}_p \right) {b}_{\boldsymbol{\theta}} \nonumber \\
	&\quad + P_{\boldsymbol{\theta}}^{-1}s \left(s \int_0^1 \exp((1-u) s P_{\boldsymbol{\theta}})V \bm{e}_j H \exp(u s P_{\boldsymbol{\theta}}) \, du \right) \bm{b}_{\boldsymbol{\theta}} \nonumber \\
	&\quad + P_{\boldsymbol{\theta}}^{-1} \left( \exp(s P_{\boldsymbol{\theta}}) - \mathbb{I}_p \right) \left( V {e}_j \boldsymbol{\kappa}(t) - V {e}_j H \bm{y}(t) \right).\label{derm}
\end{align}
Taking into account that the minimization problem relies on ${T_c}$ observations for each of  $n$ clone-type scenario,  we define the  $p\times r$-dimensional matrix $\bm{\xi}_{ci}$ such that the $j$-th column is the evaluation of \eqref{derm} at $(t_{ci},\bm{Y}_{ci})$. The gradient of the objective function with respect to \(\boldsymbol{\theta}_j\) is the $r$-dimensional vector 
\begin{equation*}
	\frac{\partial f_{ci}(\boldsymbol{\theta})}{\partial \theta_j} = -  \big[\bm{Y}_{ci}-\bm{m}_{ci}\big]^T\bm{\xi}_{cij}.
\end{equation*}

Using this expression, the variances of the estimated rates $\hat{\boldsymbol{\theta}}$, which we approximate with the diagonal of $\displaystyle {\mathcal {I}}(\hat{\boldsymbol{\theta}} )$ in \eqref{eq:sdapprox}, are given by
\begin{equation*}
\mathbb{V}[\hat{\boldsymbol{\theta}}]=\text{diag}\bigg(\sum_{c=1}^{n}\sum_{i=1}^{T_c}    (\bm{\xi}_{ci})^T \big[\bm{Y}_{ci}-\bm{m}_{ci}\big]\big[\bm{Y}_{ci}-\bm{m}_{ci}\big]^T \bm{\xi}_{ci}\bigg)^{-1}\bigg|_{\boldsymbol{\theta}=\hat{\boldsymbol{\theta}}}.
\end{equation*}
The derivations are checked numerically in the Appendix \ref{SM:sd}.

\section{Simulation study} \label{sec:5}
In this section, we present several simulation studies to assess the performance of the proposed method. 
The experimental setup, based on the system described in section~\ref{example}, reflects conditions typical of numerous biological applications.
 In section~\ref{sec:compLLA}, we evaluate the proposed method by varying the width between the observations $\Delta t$ and the number of time points $T$. We compare the results with an alternative local linear approximation method. For short time steps, we expect the local linearization to be a serious competitor, whereas for large time steps the nonlinearity of the system will make our inferential scheme preferable. In section~\ref{sec:time}, we study how accuracy and computational time vary by increasing the number of nodes and reactions.
Finally, in section~\ref{sec:xu} we compare our approach with a method-of-moments formulation introduced by \cite{xu2019statistical}, which relies on matching model-derived and empirical correlations in cell type dynamics.

\subsection{Performance as function of $\Delta t$ and $T$}\label{sec:compLLA}

In this section, we compare the performance of the LMA inferential procedure across different time steps $\Delta t$ and number of time points $T$. We use the Gillespie algorithm to simulate trajectories from the cyclic reaction system defined in section~\ref{example} with rates ${\boldsymbol{\theta}}=(0.2,   0.1,   0.2)$ and initial concentration $\textbf{y}_0=(100,100,100)$. Figure \ref{fig:trajectory} shows an example of trajectories for one of the simulations.
\begin{figure}[tb!]
	\centering
	\includegraphics[scale=0.34]{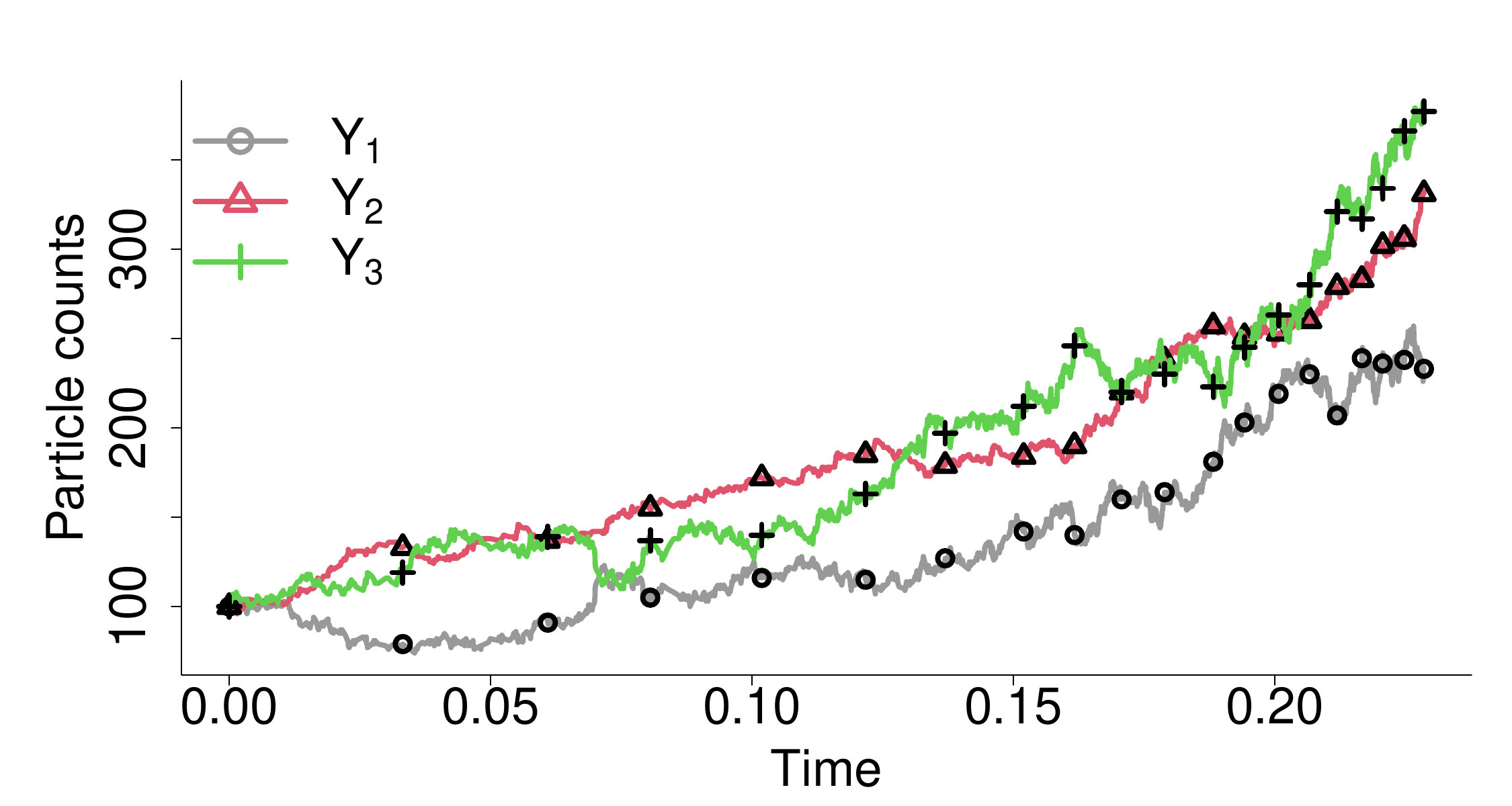}
	\caption{A trajectory generated using the Gillespie algorithm for the model described in section~\ref{example}. The symbols indicate the $T=20$ observations between every 100 simulated states (average $\Delta t \approx 0.012$).}
	\label{fig:trajectory}
\end{figure}

In order to evaluate the dependence of the results on $\Delta t$, we extract measurements from the simulated trajectories. In order to get, on average, increasing $\Delta t$ values,  we consider 5 different cases where we retain every 10, 30 50, 70, 100 values, respectively, from each trajectory for a total of $T=20$ measurements. See Figure \ref{fig:trajectory} for an example of the last case.  This leads to 5 different values for the average time steps ($0.003, 0.007, 0.009,.010,0.012$). In a second simulation, we fix the case $\Delta t = 0.012$ and consider 5 different values for the overall number of time points $T$. The simulations are repeated 100 times. 
\begin{figure}[tb!]
	\includegraphics[scale=0.29]{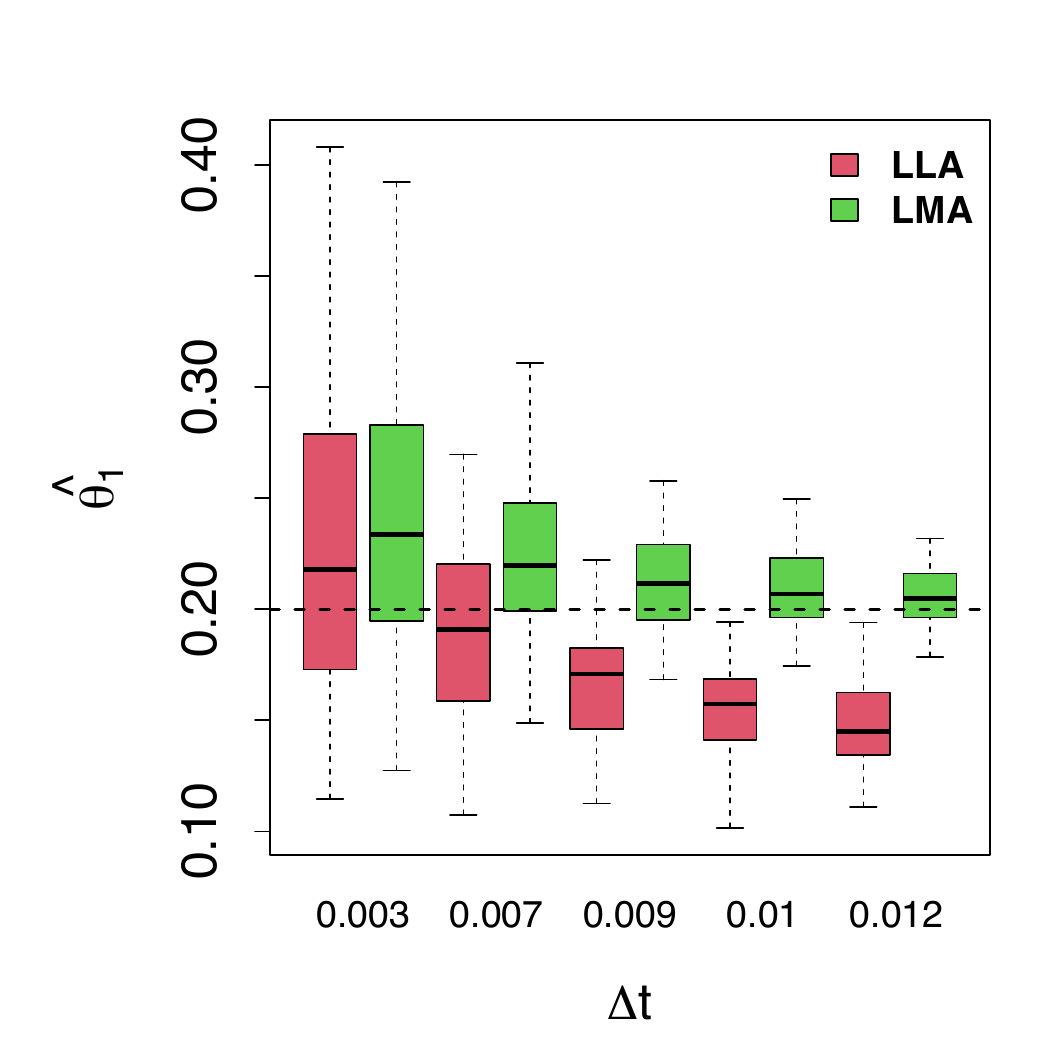}	
	\includegraphics[scale=0.29]{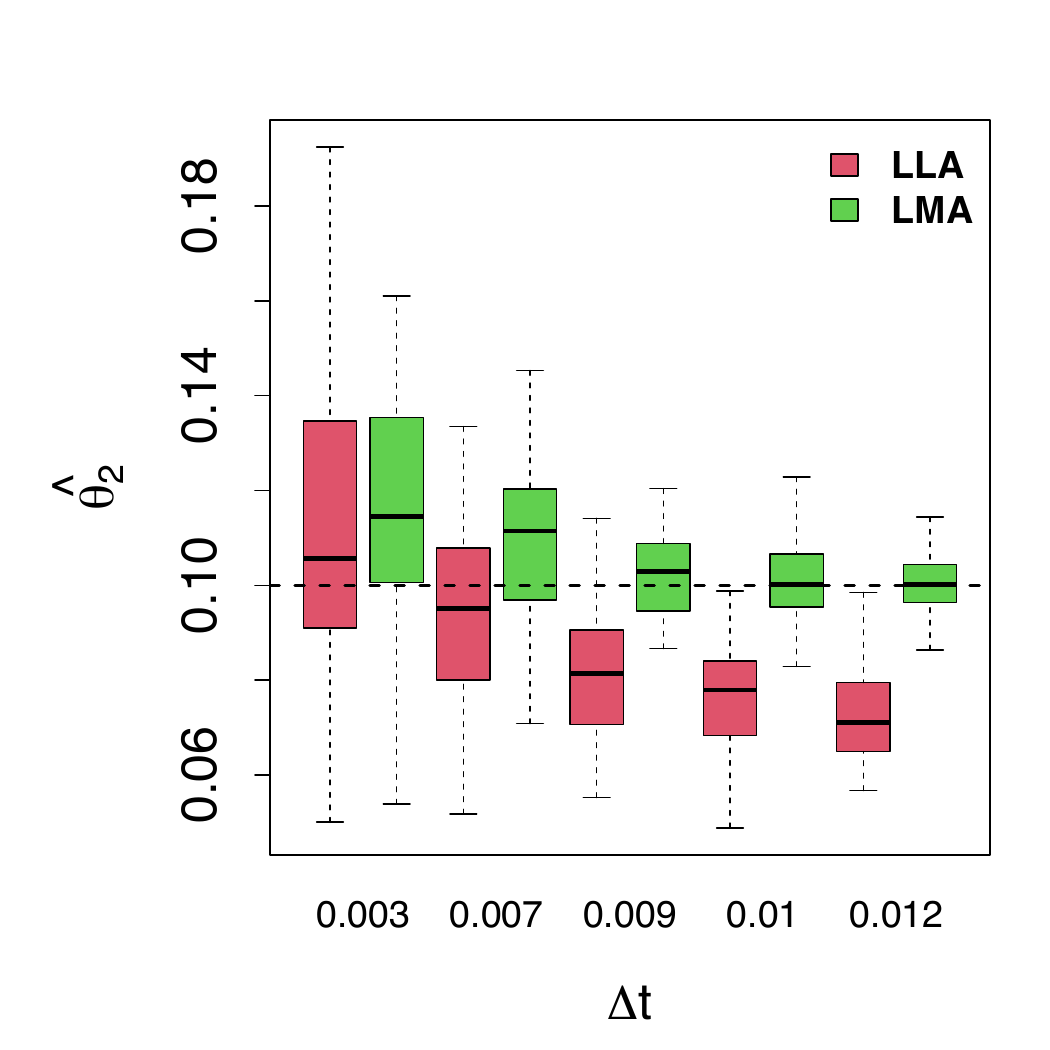}
	\includegraphics[scale=0.29]{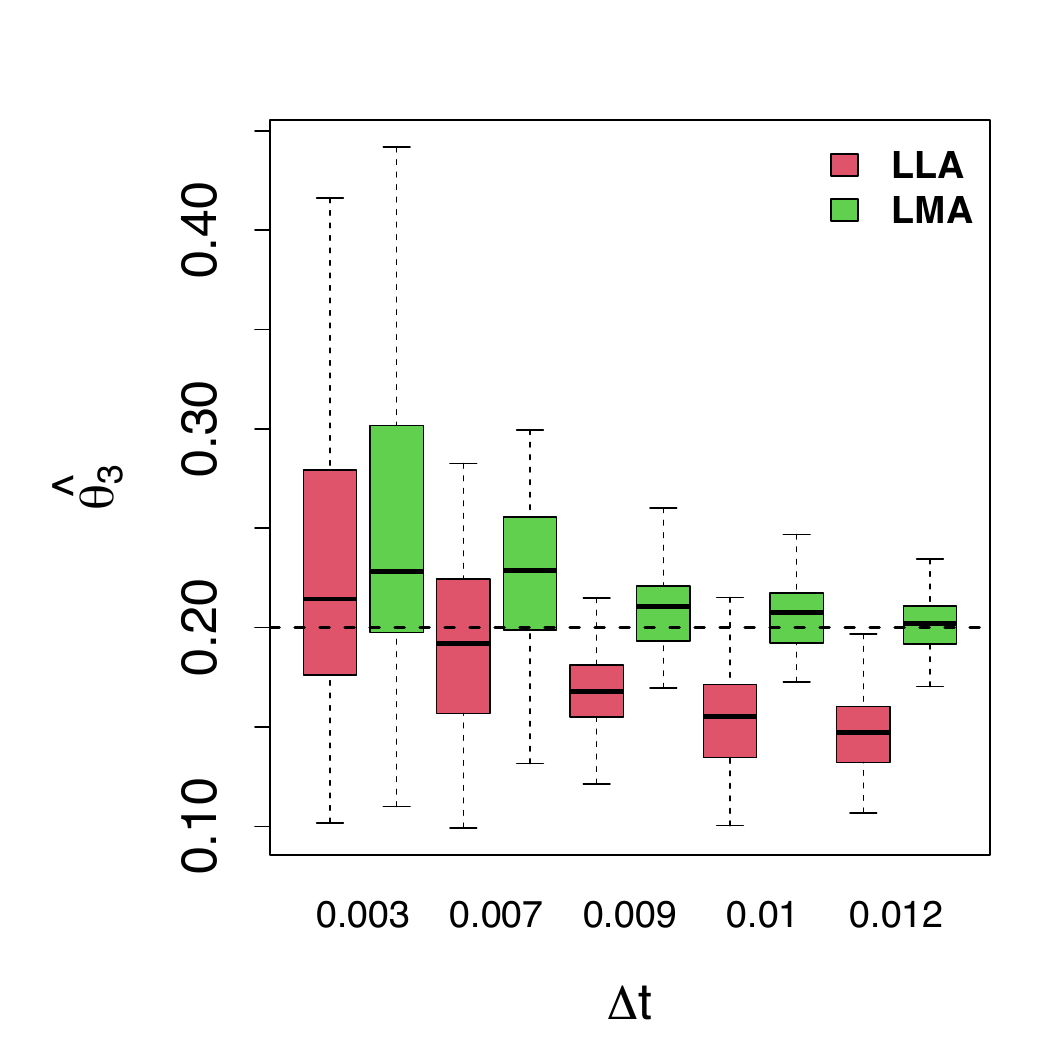}\\
	\includegraphics[scale=0.29]{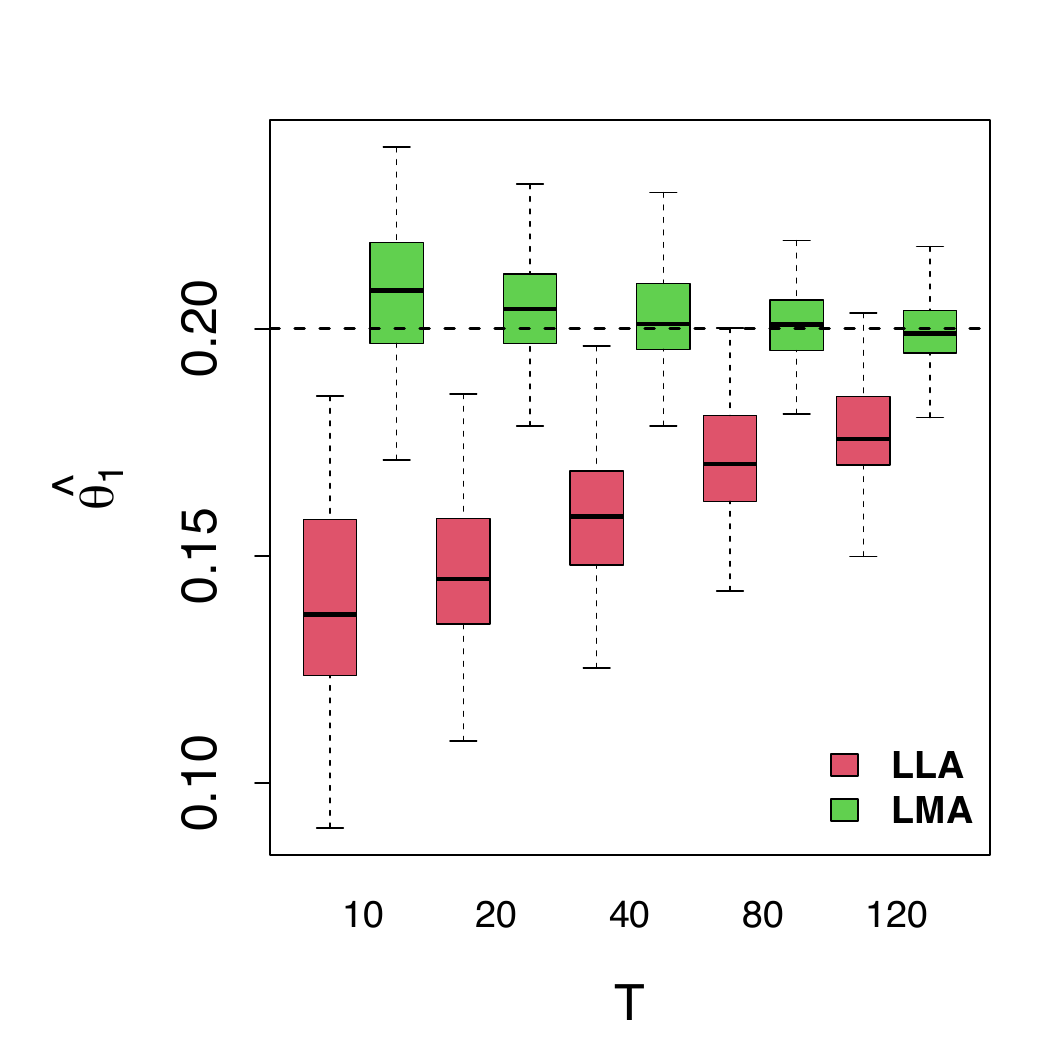}	
	\includegraphics[scale=0.29]{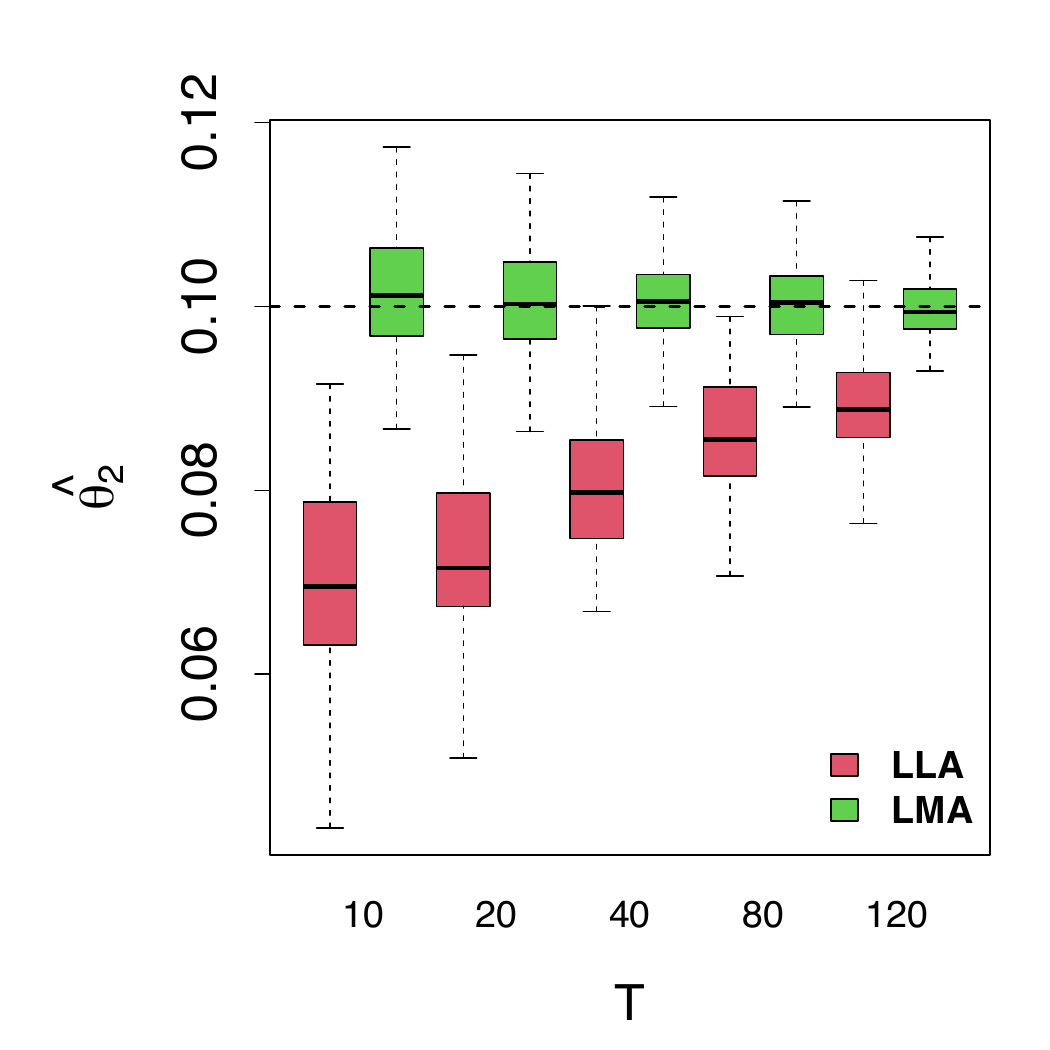}
	\includegraphics[scale=0.29]{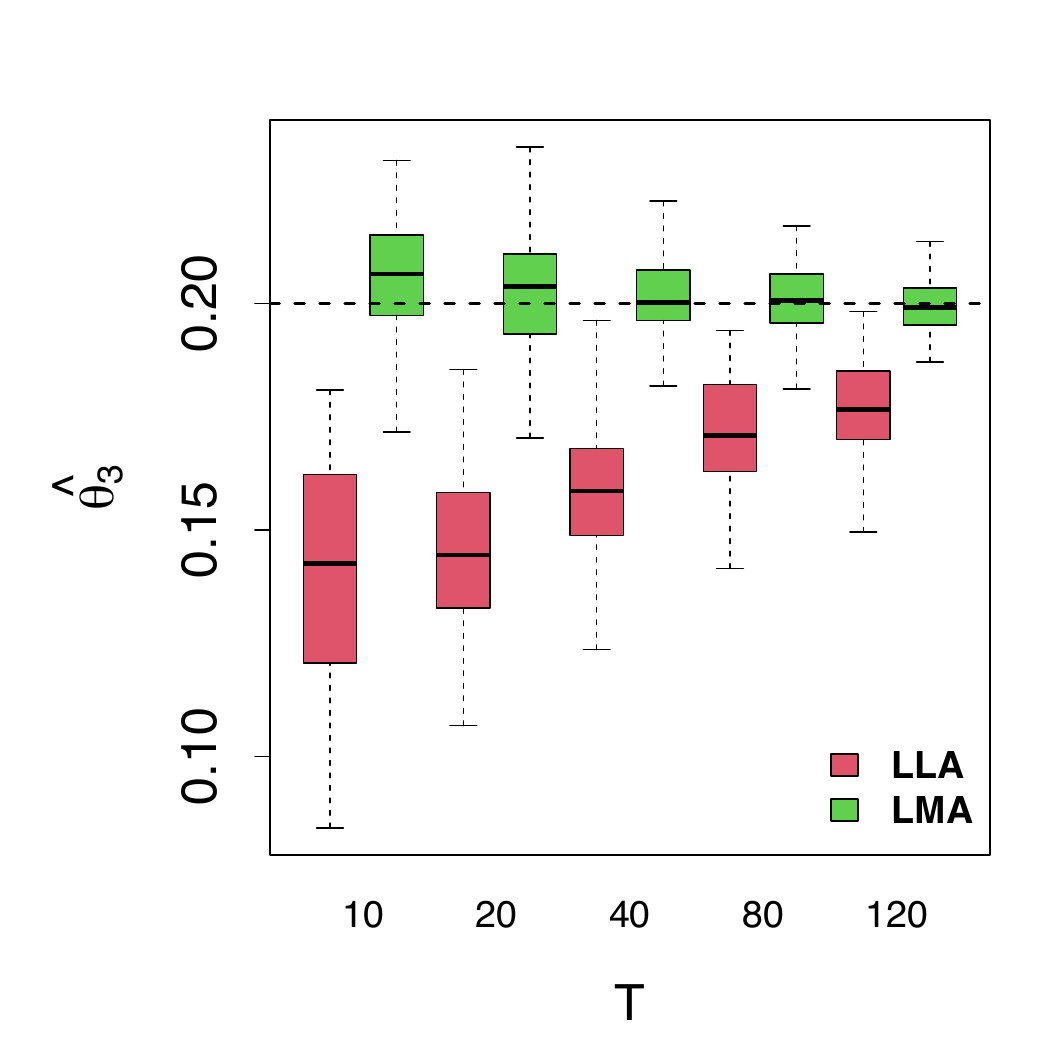}
	\caption{Distributions of the estimated rates from the local linear approximation (LLA, red) and local mean-field approximation (LMA, green) methods across $100$ simulations. Top: Keeping $T=20$ fixed, as the time step $\Delta t$ increases, the LLA estimates, unlike LMA, show increased bias. Bottom: Keeping $\Delta t =0.012$ fixed, increasing the number of time points $T$ shrinks the standard error of the LMA method in a square-root fashion. }
	\label{fig:4}     
\end{figure}
Figure~\ref{fig:4} shows the estimated parameters $\bm{\theta}= (\theta_1, \theta_2,\theta_2)$ for both the proposed LMA method (green) and the alternative LLA method (red). The results show that the performance of the two methods is comparable for small time intervals. However, as $\Delta t$ increases, the LLA approach shows an increasing bias, which is not observed for the LMA method. Furthermore, the standard error of the LMA method shrinks in a roughly $1/\sqrt{T}$ fashion, as the number of time points increases.

\subsection{Accuracy and computational time varying $r$ and $p$}\label{sec:time}
In this section, we examine the impact of varying the number of reactions $p$ and the number of particle types $p$ on the estimation accuracy and computational time required to estimate the parameters $\boldsymbol{\hat{\theta}}$. This analysis is conducted using the same experimental setup of section~\ref{sec:compLLA}, fixing ${T}=20$. The simulations are repeated $200$ times. 

The time complexity is shown with the median time across all simulations, along with the $25$th and $75$th percentiles. In terms of accuracy of parameter estimation, since the parameter values vary significantly in magnitude, we consider the Wasserstein 1-distance \citep{duy2023exact} between the distribution of an estimated parameter and the corresponding true parameter. In particular, this is given by 
\begin{equation*}
	W_1(\boldsymbol{\theta}, \hat{\boldsymbol{\theta}}) = \sum_{j=1}^r \left| F_{\theta_j} - F_{\hat\theta_j}\right|,
\end{equation*}
where  $F_{\hat\theta_j}$ is the empirical cumulative distribution of $\hat\theta_j$ from the $200$ simulations, while $F_{\theta_j}$ is the degenerate distribution on the true value $\theta_j$.  The median and interquantiles of this measure are calculated by considering $1000$ bootstrap versions of the $200$ datasets and calculating the Wasserstein 1-distance from each of these.

In a first simulation, we vary the number of particle types $p\in\{3,6,9\}$. For $p=3$, we consider the cyclic reaction system in section~\ref{example}. In order to still have a $P_{\boldsymbol{\theta}}$ invertible, we generate the settings with a higher $p$ by simply repeating the same system a number of times, as shown in Figure~\ref{tabp} of Appendix~\ref{tabels}. The results of the simulations are shown in Figure~\ref{fig:results}.  The top-left panel shows a super-linear dependence of computational time on $p$. The bottom-left panel shows that the accuracy of the estimates decreases linearly with respect to the number of $p$ states.
\begin{figure}[tb!]
	\centering
	\includegraphics[scale=0.3]{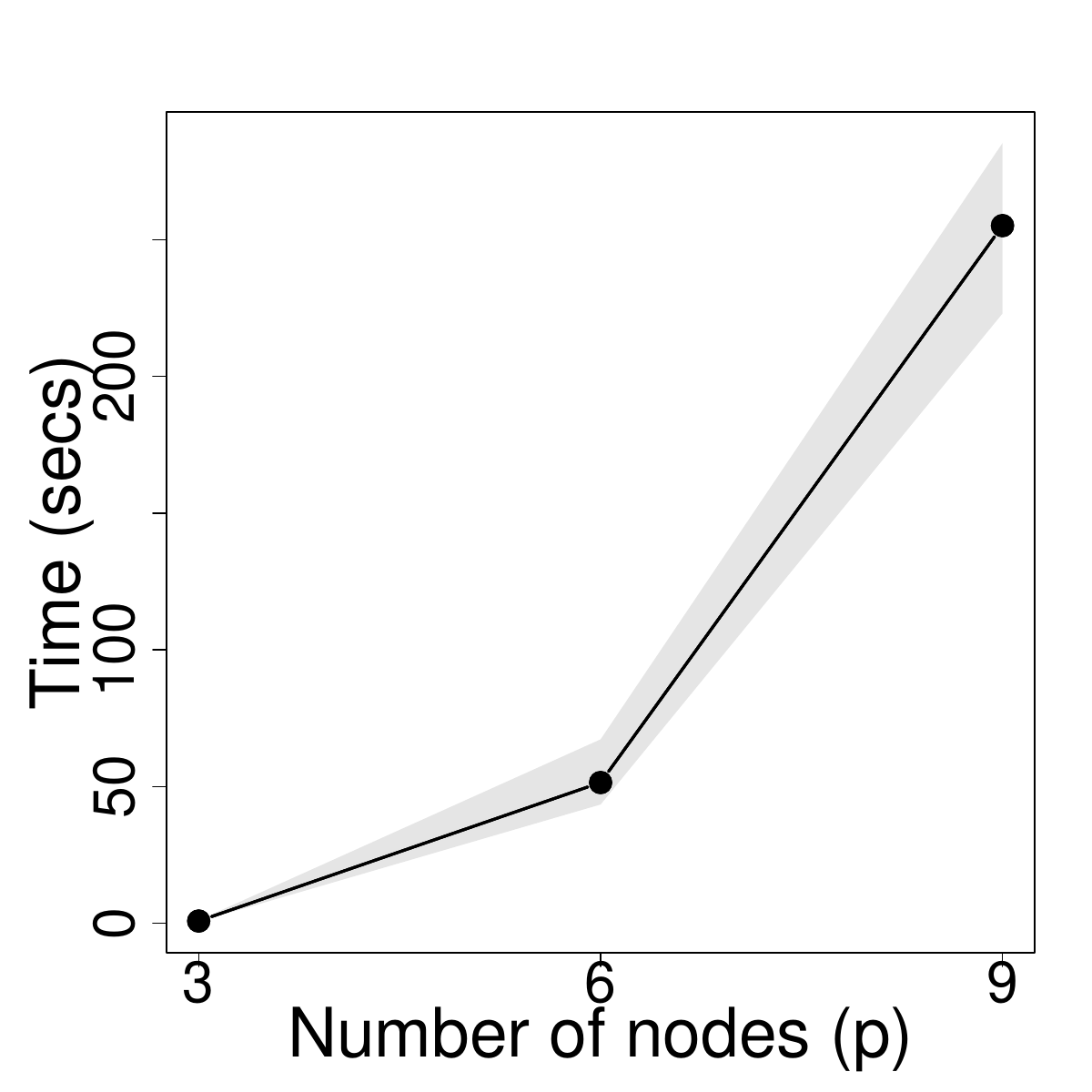} 
	\includegraphics[scale=0.3]{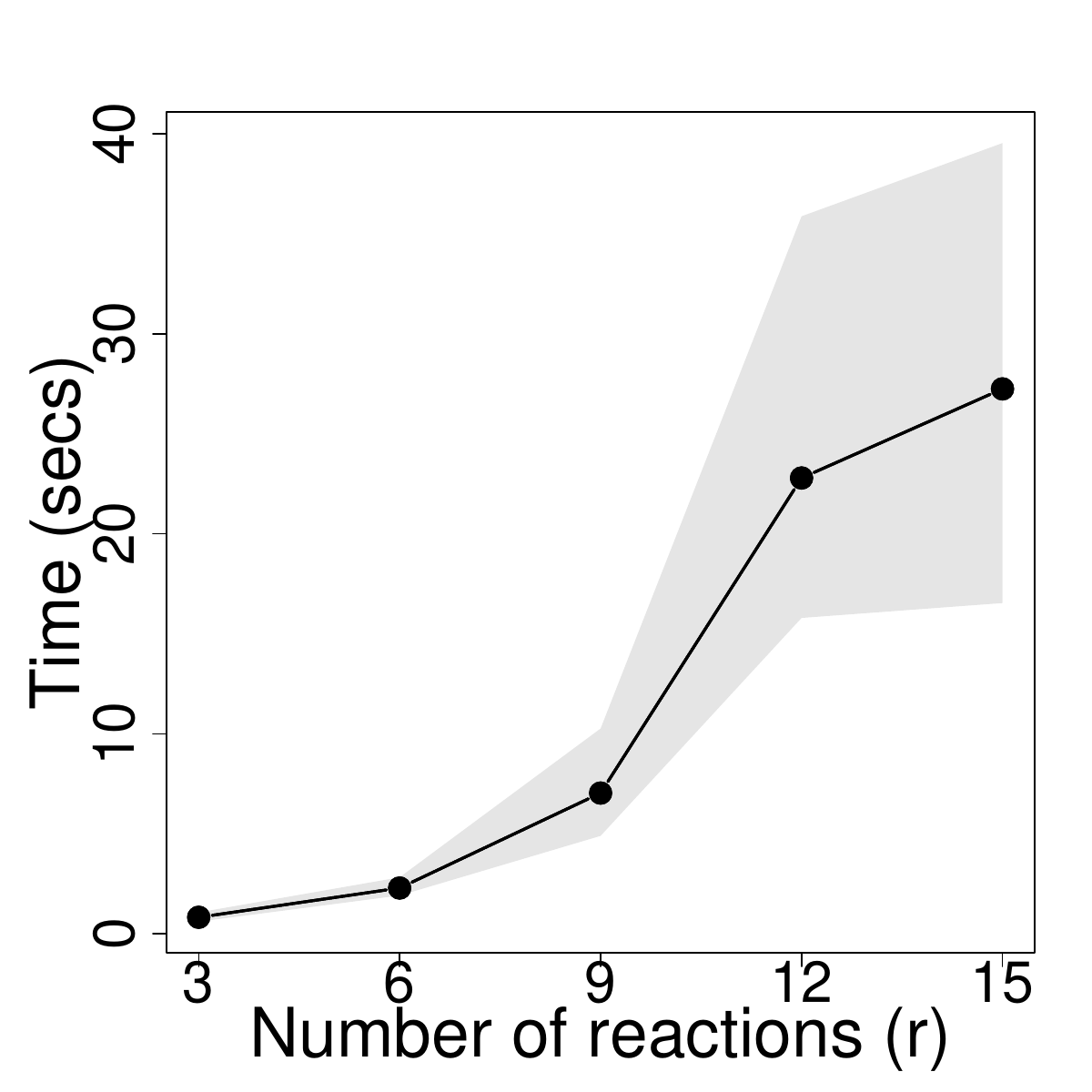}
	\includegraphics[scale=0.3]{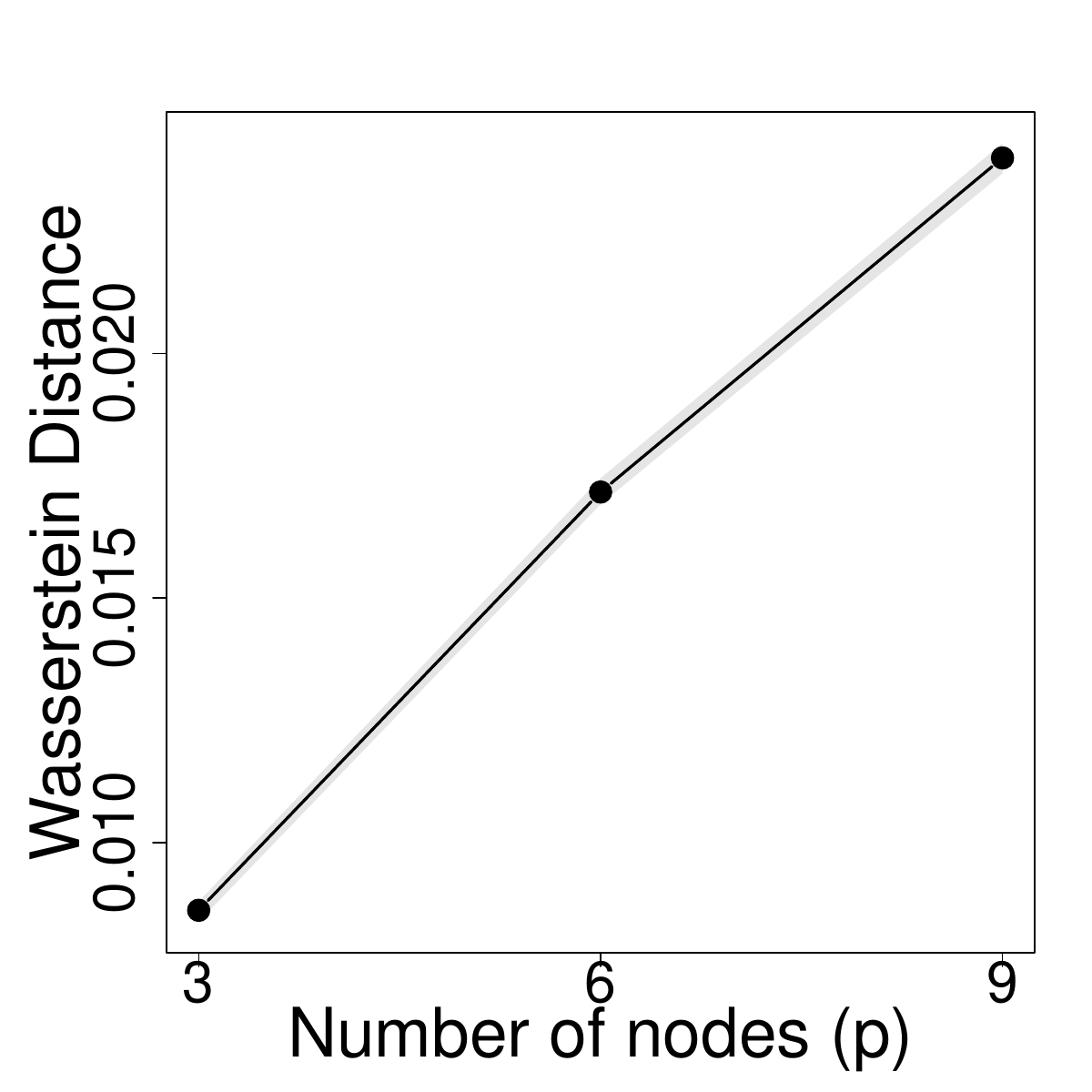}
	\includegraphics[scale=0.3]{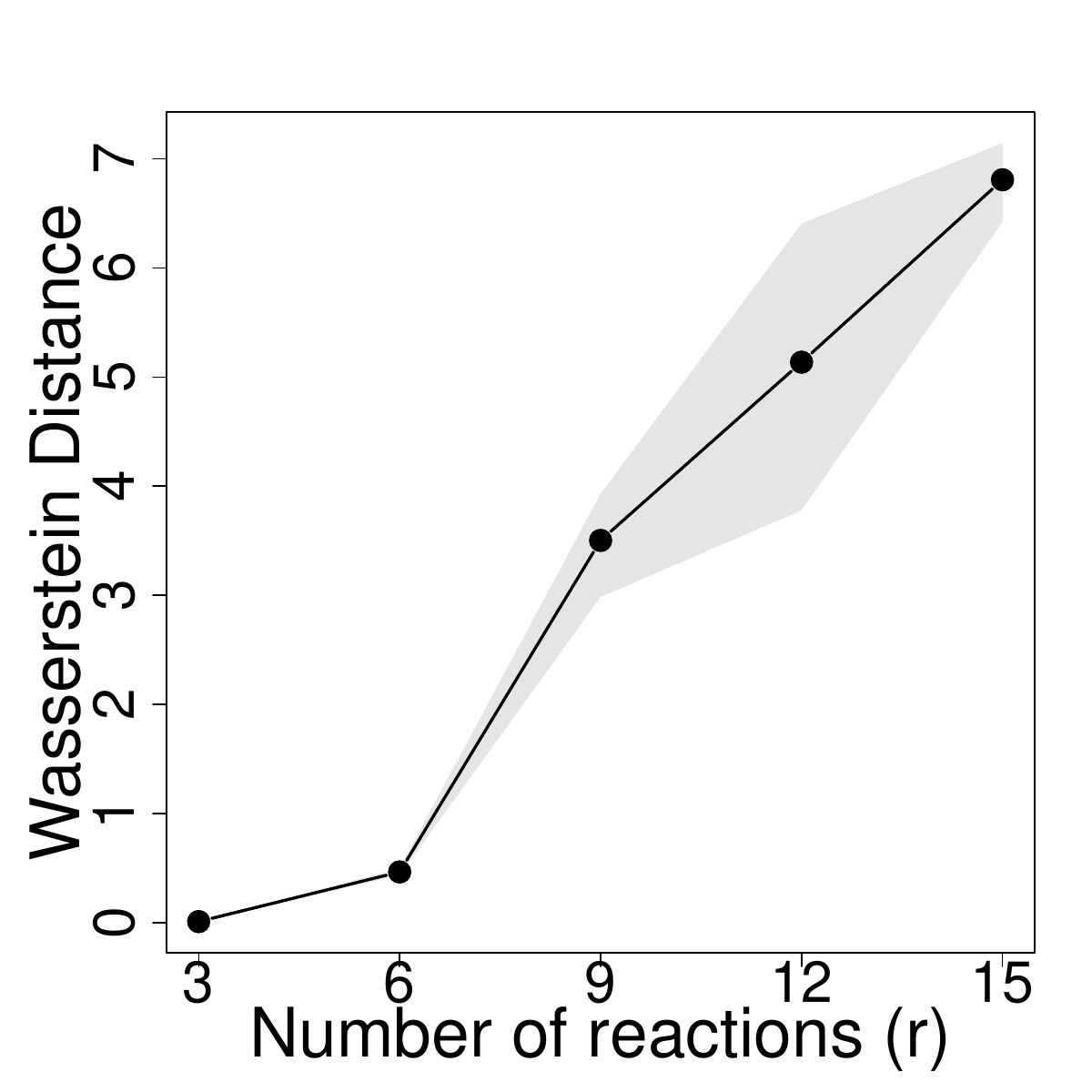}
\caption{Median computational time (top) and Wasserstein 1-distance (bottom) of the LMA algorithm, as a function of the number of reactions ($r$, right) and states ($p$, left). The shaded area represents the interquartile range across 200 simulations.}
	\label{fig:results}
\end{figure}

In a second simulation, we fix the number of particles at $p=3$ while progressively varying the number of reactions $r \in \{3, 6, 9, 12, 15\}$. We generate the different systems with varying $r$ by following the scheme in Figure~\ref{tabr} of Appendix~\ref{tabels}. The top-right panel of Figure~\ref{fig:results} shows a quasi-cubic dependence of computational time on $r$, which aligns with the theoretical predictions discussed in section~\ref{stiffness}. The bottom right plot shows how the overall accuracy of the estimates improves approximately linearly with the increase in the number of reactions.

\subsection{Comparison with M-estimator by  \cite{xu2019statistical}}\label{sec:xu}

Finally, we present a simulation study comparing our proposed method with an alternative method proposed by \cite{xu2019statistical}. The authors considered the empirical correlation between the dynamics of the particle types with the theoretical one defined by the chemical master equation. They applied their method to a human cell differentiation system involving hematopoietic stem cells (HSCs), two progenitors, and five mature cell types. The system is described in Figure~\ref{fig:Xua}. To obtain an analytical solution for the evolution of the moments, the authors assumed linear propensity functions with respect to the state concentrations. This concept is strongly restrictive in the description of cell dynamics as it implies only exponential growth and extinction of the cells, as also noted by \cite{pellin2023tracking}. However, in order to properly compare the two methods, we have kept the same assumption. Note that no changes to the method are necessary as it is a generalization of the case of linear hazard functions.  

As HSCs and progenitor cells are latent in the reference method, we only identify the dataset with mature cells after generating the complete data, in order to maintain the same original setting. In terms of reaction rates, we multiply the original values defined in \cite{xu2019statistical} by 1000. This merely constitutes a time unit change. We set the HSC duplication parameter to $\lambda=2850$, and the differentiation in progenitors a-type to $\nu_a =1400$, and b-type  to $\nu_b =700$, respectively. Such cells have death rates of $\mu_a=50$ and $\mu_b=40$. The type-a progenitor can differentiate into Granulocytes with rate $\nu_1=3600$ and into Monocytes with rate $\nu_2=1800$. The type-b progenitor differentiates into T-cells with rate $\nu_3=1000$, into B-cells with rate $\nu_4=2000$ and finally into Natural-Killer (NK) cells with rate $\nu_5=1200$. The death rates of the 5 mature cells are given by $\mu_1=26,\mu_2=13,\mu_3=11,\mu_4=16,\mu_5=9$, respectively. 

Each simulation starts from one HSC and no other cell types.  From this, $n=100$ replicates or clones, are simulated in $T=5$ time steps each. The observation times are stochastically defined, but in order to maintain the biological significance of blood sampling at defined times, the indexed mean time is defined for all clones, resulting in a data span $t=(0.00,0.08,0.11,0.14,0.17)$. We apply our proposed LMA approach and the method of \cite{xu2019statistical} to these simulated data. As in \cite{xu2019statistical}, the death rates are kept fixed, as these values are taken from the biology and immunology literature, while all other parameters are estimated from data. Given that the  \cite{xu2019statistical} algorithm relies on an initial value, we perform a sensitivity checking by restarting the inference procedure $100$ times for each simulation. The solution corresponding to the minimum cost function among the others is then selected as the final value. 

The results of the simulation study, the steady-state distribution process and the parameter distributions with $300$ simulations are shown in Figure~\ref{fig:Xu}. 
\begin{figure}[t!]
	\centering
	\subfloat[]{
		\includegraphics[scale=0.15]{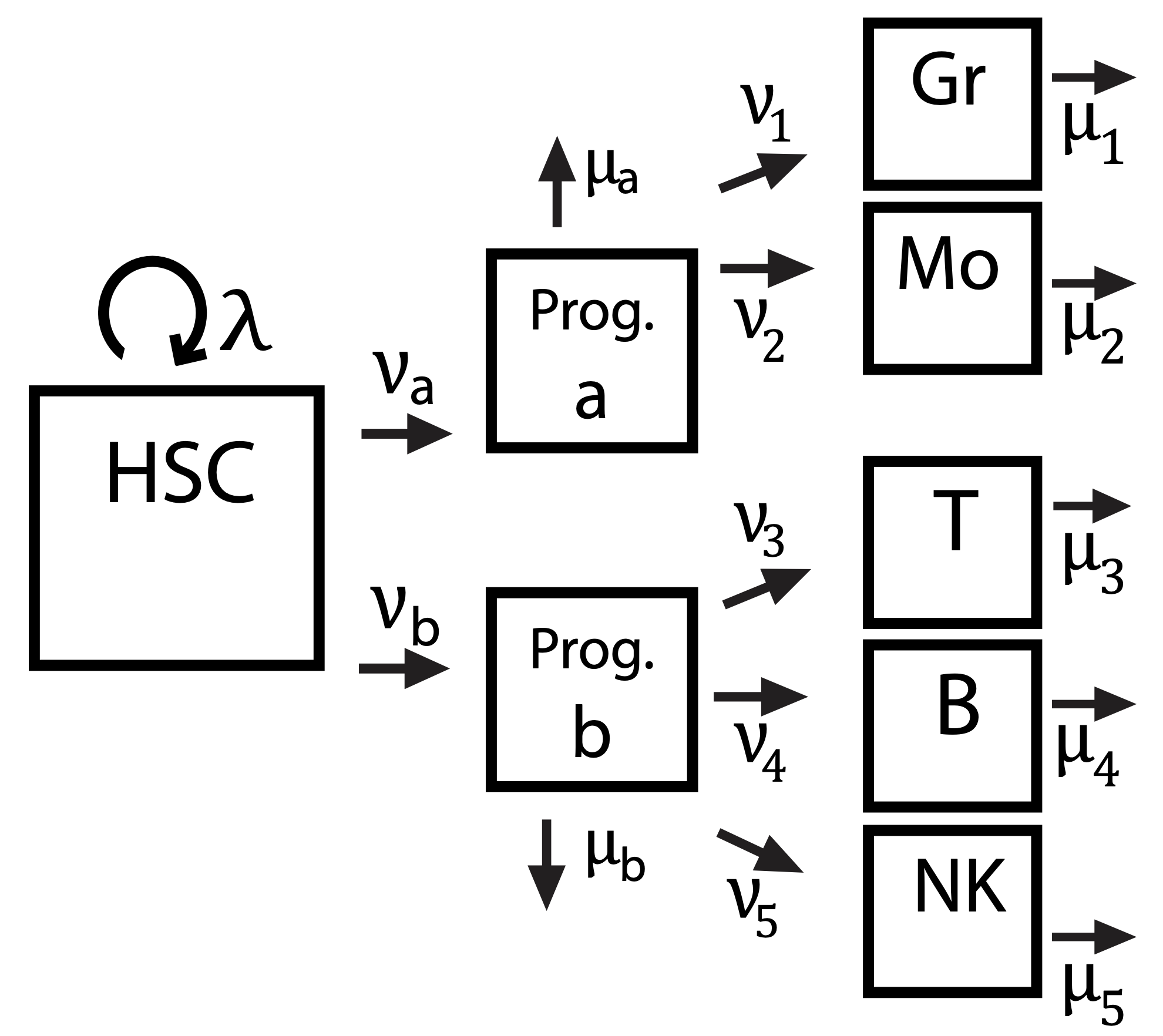}
		\label{fig:Xua}
	}
	\hspace{1.5cm}
	\subfloat[]{
		\includegraphics[scale=0.15]{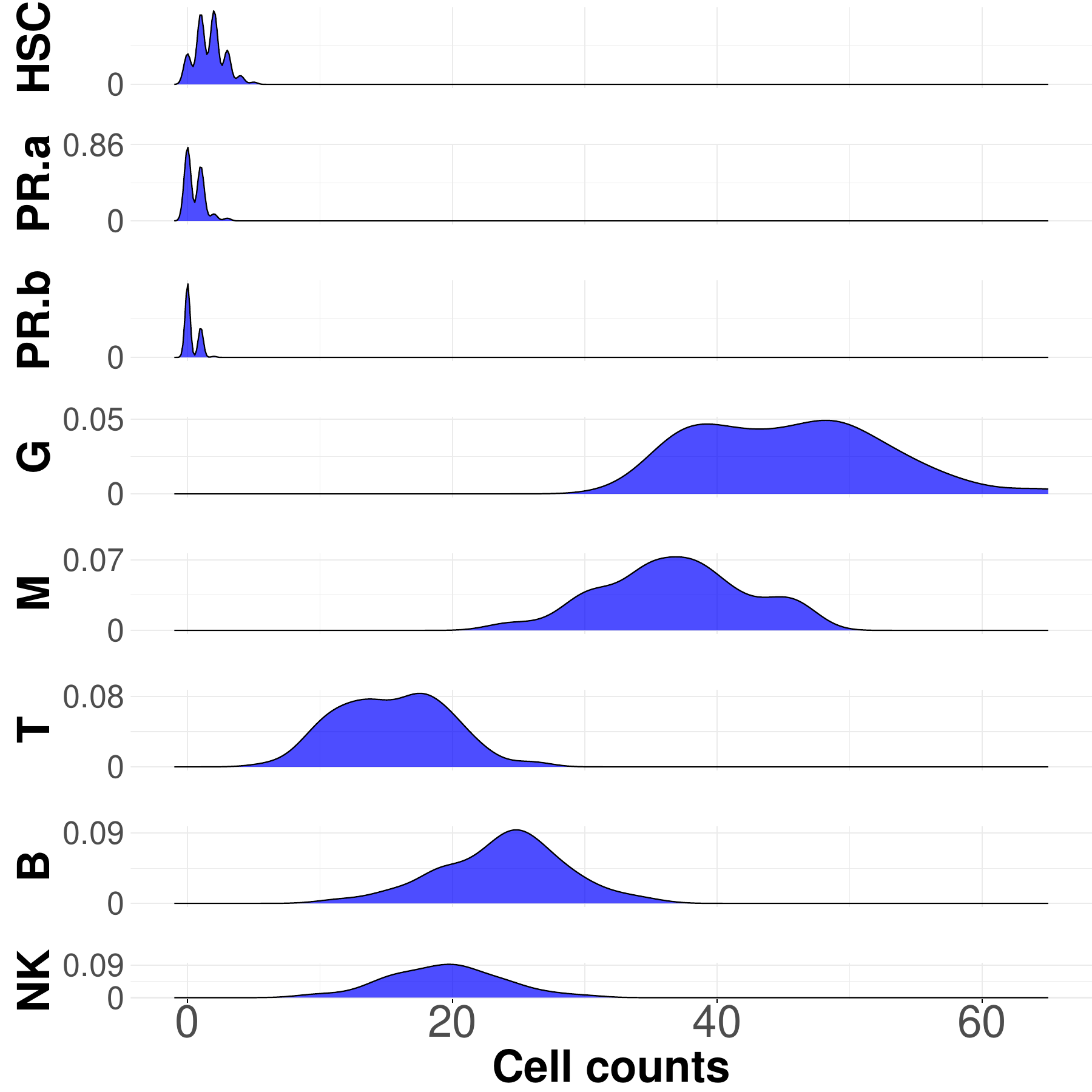}
	}\\  % This line moves the next figure to a new line
	\subfloat[]{
		\includegraphics[scale=0.45]{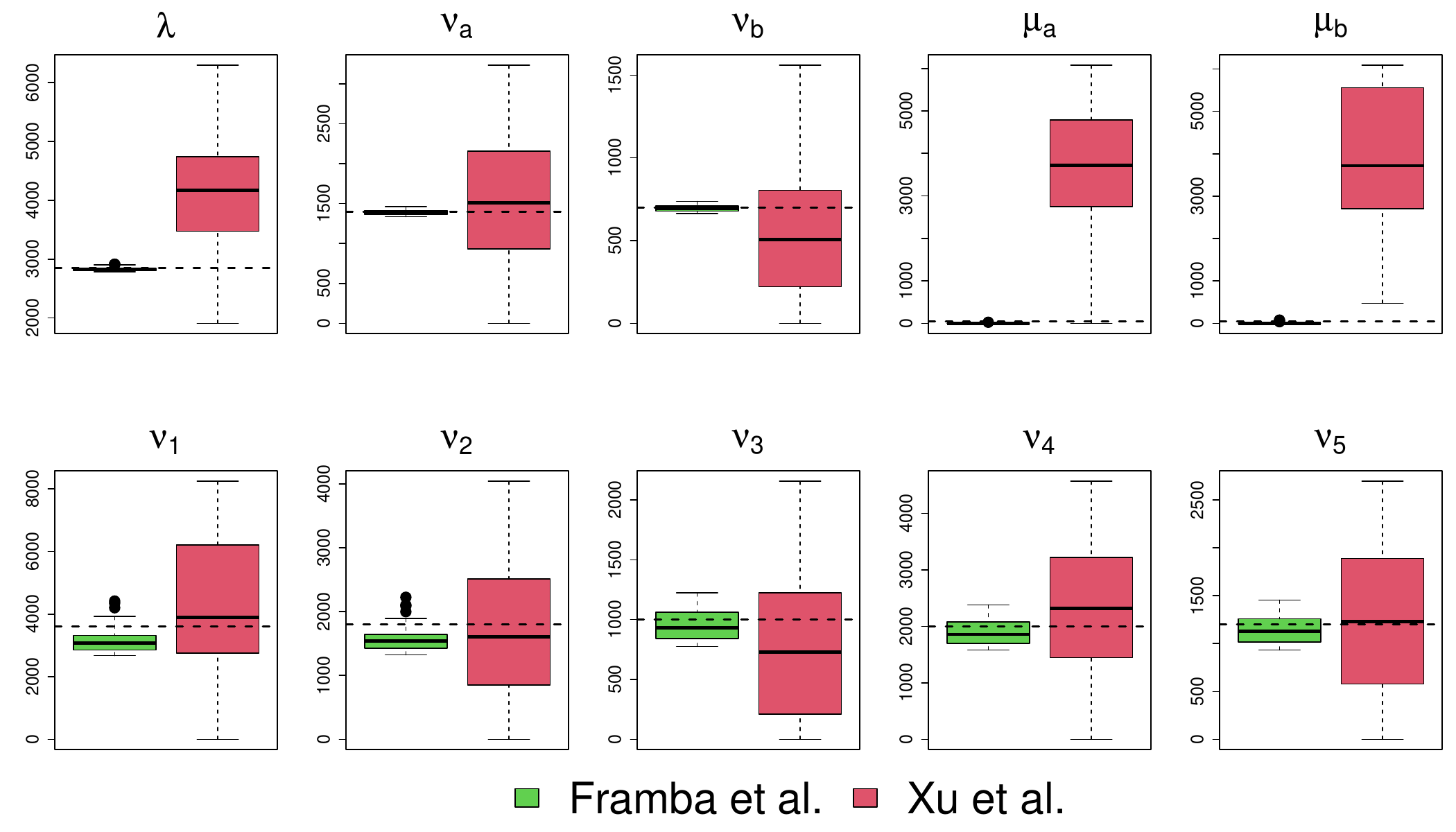}
	}
	\caption{Comparison with \cite{xu2019statistical} correlation-based M-estimator: a) Blood cell differentiation scheme. b) Cell types multi-modal steady-state distribution calculated using 100 clones. c) Boxplots of the estimated parameter distributions from 100 simulations. The 10 unknown rates are shown, the true values are indicated by a horizontal dashed red line. The proposed method is unbiased and more accurate than the M-estimator by \cite{xu2019statistical}.}
	\label{fig:Xu}
\end{figure}
The proposed LMA approach outperforms  \cite{xu2019statistical} in several aspects. Firstly, most of the LMA parameter estimates, with the possible exception of $\hat\nu_1$ and $\hat\nu_2$, are unbiased, in contrast to the estimates obtained with the  \cite{xu2019statistical} method. Secondly, the precision of the LMA estimates is significantly higher than that of  the \cite{xu2019statistical} estimates. The reason for these improvements is that the method in \cite{xu2019statistical} is based on matching second order moments, which are inherently less stable than first order methods.

\section{Cell lineage barcoding in rhesus macaques} \label{sec:rhesus}
Clonal tracing modeling is commonly used in genetic studies to enhance our understanding of blood cell formation, also known as hematopoiesis. The hematopoietic process is represented as a tree structure with a self-renewing hematopoietic stem cell at its origin. This cell evolves from a pluripotent state into mature blood cells through several intermediate progenitor stages. The regulation of the number of circulating cells in the blood is maintained through the rates of birth, death, and differentiation. An important challenge in studying this system in living organisms is that only mature cells are observable and can be accessed through blood sampling. To gain further insights, recent research has analysed cell production in non-human primates, which closely mirrors human physiology due to the similar lifespan and frequencies of hematopoietic stem progenitor cells (HSPCs) \citep{kim2000many}. In this article, we considered an in-vivo clonal tracking dataset on Rhesus Macaques \citep{wu2014clonal} and aim to recover the rates of birth, death, and differentiation from these data. 

\paragraph{Hematopoietic stem cell gene therapy.}
In the \citep{wu2014clonal}  gene therapy clonal study, unique DNA barcodes IDs were first introduced into autologous CD34+ HSPCs  utilising a high-diversity lentiviral barcode library. Then, such cells were reinfused into the three myeloableted animals \citep{shepherd2007hematopoietic}. Once reinfused, the genetically modified HSPCs home to the bone marrow, where they engraft and begin to repopulate the haematopoietic system. These barcoded HSPCs proliferate and differentiate into various blood cell lineages, allowing the tracking of individual clones over time. The integrated barcodes remain stable within the genome of the progeny cells, enabling the detailed monitoring of clonal dynamics, lineage contribution, and the longevity of specific clones across different haematopoietic compartments.   In the \cite{wu2014clonal} study, samples from peripheral blood, bone marrow, and lymph nodes for the three monkeys were collected monthly, enabling the identification of unique clonal patterns and contributions to different cell types, such Granulocytes (G),  Monocytes (M),  T, B, Natural-Killer (NK) cells. The entire observation time varies from subject to subject and corresponds to 4.5 months for monkey ZG66, 6.5 months for monkey ZH17, and 9.5 months for monkey ZH33. Each barcoded lineage $\bm{Y}(t)$  is assumed to be an independent realization of the hematopoietic stochastic model.

The data were imported from the \texttt{Karen} library  \citep{del2022stochastic}. %Further details on dataset acquisition, the transduction protocol, and culture conditions can be found in the original article by \cite{wu2014clonal}, while a comprehensive analysis of clonal models is provided in \cite{espinoza2021interrogation}. 
The dataset contains many missing sampling events. As a pre-processing step, we exclude the time points when no barcodes were detected as well as all clones with less than 3 temporal observations. This leads to a total of $555$ unique barcodes IDs, which are split between  $434$, $50$, and $19$ different clone-types in specimen ZG66, ZH17 and ZH33, respectively.
\begin{figure}[t!]
	\centering
	\includegraphics[scale=0.35]{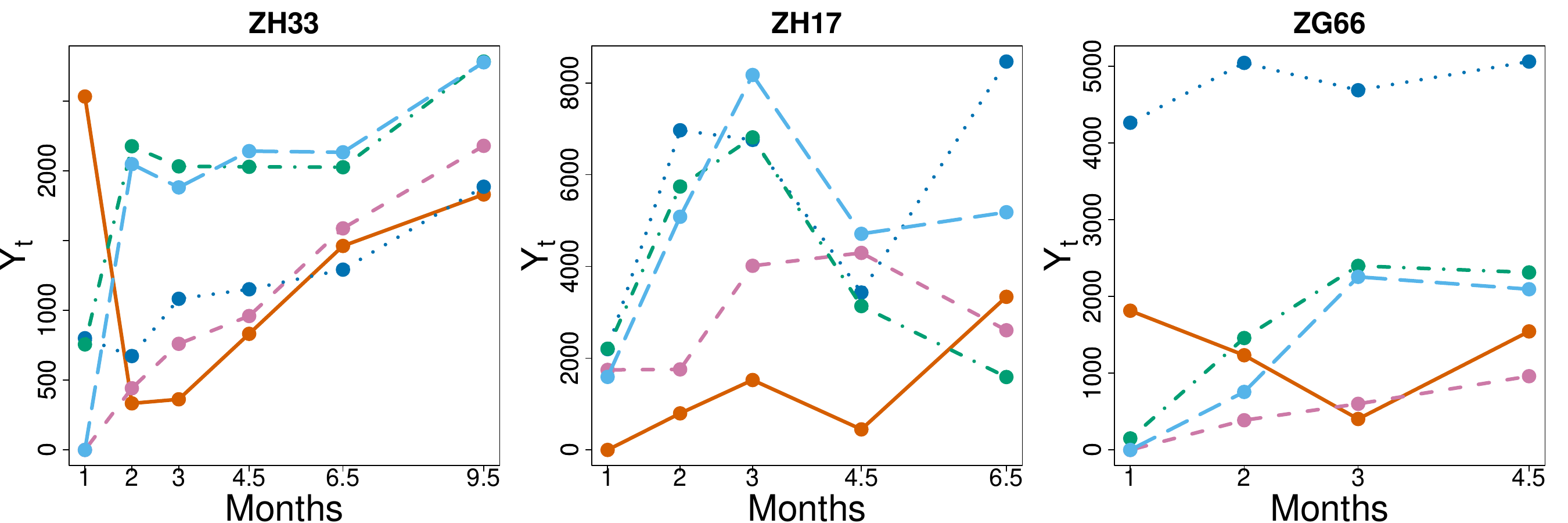}
	\centering
	\includegraphics[scale=0.34]{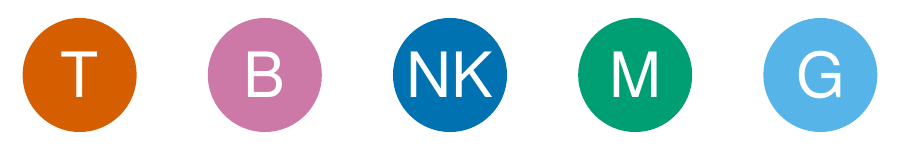}
	\caption{Mean concentration over time of each Rhesus Macaque specimen following transplantation. The $p=5$ cell types are reported with different colours and line styles.}
	\label{fig:clones}
\end{figure}
\paragraph{Hematopoietic reaction network selection.}
Figure~\ref{fig:clones} shows the average differentiation trajectories for each specimen. We posit that the dynamics of cell differentiation are universal and thus independent of individual subjects. Consequently, we assume that the differentiation process in each subject is governed by the same rate parameters. Following the same approach as in \cite{pellin2023tracking}, we define the death reaction hazard rates as a quadratic function of the underlying abundances, representing a natural saturation effect.   

Using the proposed LMA approach, we compare competing models of hematopoiesis from the available experimental data. To this end, we define the net matrix of the full model  consisting of 30 reactions, of which 10 are birth and death reactions and 20 are all possible differentiations from one cell to another.  We initialize the reaction rates with the local linear approximation estimates. Although previous similar studies fix the death parameters according to established chemical regimes \citep{hellerstein1999directly,xu2019statistical}, we keep such rates variable. We then search through the space of possible models, going from an initial model $m_1$ consisting of only one reaction out of the 30 possible ones to the full saturated model. For each model, we estimate the parameters using the LMA approach, i.e., by solving the optimisation problem \eqref{prob}. At each step of a stepwise procedure, we iteratively add and subtract the reaction that most reduces the Bayesian Information Criterion (BIC). %\citep{bhat2010derivation}.
% \begin{equation*}
	%\text{BIC}(\bm{\hat{\theta}}) = N\mathcal{C}p\cdot\bigg( f(\hat{\boldsymbol{\theta}}) -   \log(N\mathcal{C}p)\bigg)+ \log(N\mathcal{C}p)\cdot r
	% \end{equation*}
This methodology leads to a sequence of models with increased complexity. The results are shown in Figure~\ref{fig:bic} (right). For each degree of complexity,  the lowest BIC of the optimal model is shown in red, while the BIC values of the unselected models are shown in gray. 

The optimal model and complexity level, i.e., the model corresponding to the lowest overall BIC, contains 10 reactions. Table \ref{app:estimate} reports the reactions that define this model, the corresponding reaction rates estimated by the proposed LMA method, together with the standard errors calculated as described in section~\ref{sec:sd}. Due to the quadratic nature of the death rates and to particle counts in the order of $10^3$ (Figure~\ref{fig:clones}), we can see how reactions involving $M$ and $B$ tend to be the slowest (with rates in the order of $10$ days), reactions involving $G$ and $T$ occurr at a rate of $1$ day, while death rates occur at the fastest rate of $10^{-1}$ days. Standard errors are small with respect to the size of the parameters, suggesting small uncertainty on the estimates.
\begin{table}[h!]
	\centering
	\begin{tabular}{ccc} \hline
		reaction &  reaction rate estimate & Standard errors \\ \hline
		$G\to \emptyset$ & $3.054 \cdot 10^{-7}$ & $2.405 \times 10^{-9}$\\
		$NK \to \emptyset$  & $1.759 \cdot 10^{-7}$ & $3.236 \times 10^{-10}$\\
		$M \to G$ & $1.340 \cdot 10^{-2}$ & $3.593 \times 10^{-4}$ \\
		$B \to M$  & $1.183 \cdot 10^{-2}$ & $4.916 \times 10^{-4}$\\
		$M \to B$  & $1.180 \cdot 10^{-2}$ & $3.69 \times 10^{-4}$\\
		$G \to NK$  & $2.339 \cdot 10^{-3}$ & $1.935 \times 10^{-4}$\\
		$G \to M$  & $5.876 \cdot 10^{-3}$ & $4.123 \times 10^{-4}$\\
		$B \to G$ & $7.048 \cdot 10^{-3}$ & $4.690 \times 10^{-4}$\\
		$G \to T$ & $6.413 \cdot 10^{-3}$ & $4.846 \times 10^{-4}$\\
		$T\to B$  & $9.469 \cdot 10^{-3}$ & $1.950 \times 10^{-4}$\\ \hline
	\end{tabular}
	\caption{Estimated rates, and corresponding standard errors, from Rhesus Macaques gene therapy trial data, based on the optimal model of cell differentiation. }
	\label{app:estimate}
\end{table}

In the search of all models, each reaction may be included in a number of models of varying complexity. As a way of obtaining a measure \(p_j\) of overall relevance of each reaction, we determine the evidence that a specific reaction is included by summing the weights of all models that include this reaction. Formally, $	p_j = \sum_{i \in \mathcal{M}_j} w_i,$
where \(\mathcal{M}_j\) is the set of all models that contain the parameter \(\theta_j\), and $w_j$ is the rescaled BIC-based weight. This is defined by
\begin{equation}
	\label{bic_weight}
	w_j = \frac{\exp\{-\frac{1}{2}(BIC_j - BIC_{\text{min}}) \}}{\sum_{h=1}^{r} \exp \{-\frac{1}{2} (BIC_h - BIC_{\text{min}})\}    },
\end{equation}
where \(BIC_{\text{min}}\) is the minimum BIC value among all the models considered. 
The cell differentiation network in Figure~\ref{fig:bic} (left) has edges whose thickness is proportional to the $p_j$ value. 
\begin{figure}[t!]
	\centering
	\includegraphics[scale=0.3]{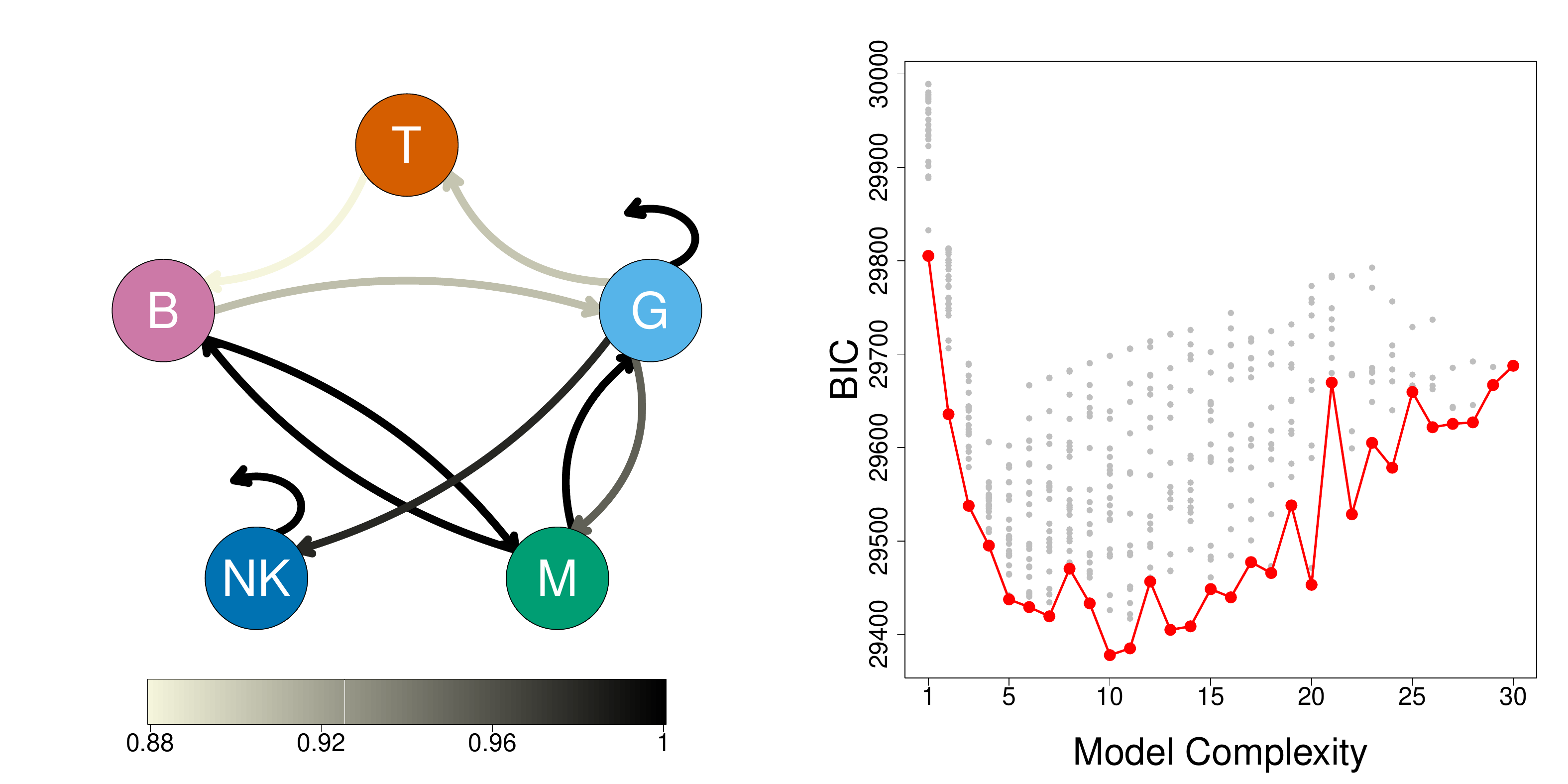}
	\centering
	\caption{BIC model selection on Rhesus Macaques data. Right: Reaction systems ordered by complexity, from the simplest (single-reaction model) to the most complex (including all possible birth, death, and single-reactant differentiation events). For each complexity level,  the BIC of the best model is shown in red, while the others are shown in gray. The optimal overall model contains $10$ reactions. Left: HSC differentiation process, where the thickness of each arrow refers to the rescaled BIC-based weights described in \eqref{bic_weight}.}
	\label{fig:bic}
\end{figure}
From this we can see, how, although each element appears both as a reactant and a product in the optimal model (Table \ref{app:estimate}), NK cells do not tend to differentiate into other cell types. Moreover, the figure shows how the less frequent edges form a loop connecting nodes B, T, and G cells, while the reaction events that appear most relevant for explaining the count trajectories involve monocytes. % while T-cells have the lowest rate of differentiation. 

\section{Conclusion} \label{sec:6}

In this study, we have developed and assessed a novel methodology for parameter inference in quasi-reaction systems, with a particular focus on situations where observations are made at large time intervals. Traditional local linear approximation methods, while computationally efficient, often fail to capture the complex nonlinear dynamics inherent in biological systems, especially when data is sparse or irregularly spaced. Our proposed approach, which extends mean-field approximation techniques, addresses these limitations by providing an explicit solution for the first moments of the state distributions under a generic quasi-reaction system. 

The performance of the proposed local mean-field approximation method is evaluated through an extensive simulation study and compared against other methods. The results demonstrate that our method significantly outperforms local linear approximation, particularly as the time interval between observations increases. Thanks to the availability of an explicit solution, the proposed method is shown to be robust to stiffness, a common occurrence in biological systems where processes operate at vastly different time scales. In addition, we illustrate the approach for the study of cell differentiation from gene therapy clonal tracking data. The approach returns an estimate of the cell populations dynamics and provides meaningful insights into the underlying biological processes. 

This work advances the inference of quasi-reaction systems by providing a versatile and reliable tool, which is suited to any generic quasi-reaction system and which can be particularly valuable for applications in compartmental studies and multi-type branching models, where traditional methods may be inadequate or computationally prohibitive.

\bibliographystyle{chicago}
\bibliography{biblio} % Nome del file .bib senza estensione

\appendix
\section{Jacobian of the hazard function}\label{jac:der}
\begin{proposition}
	Given the intensity function $
	\lambda_j(\bm{Y}(t);\boldsymbol{\theta}) = \theta_j \prod_{l=1}^p \binom{Y_l(t)}{k_{lj}},
	$ the $jl$-th element of the Jacobian matrix $\Lambda(\bm{Y}(t);\boldsymbol{\theta})\in \mathbb{R}^{r\times p}$ is given by
	\[
	\Lambda_{jl} = \theta_j \prod_{i=1}^{p} \binom{Y_i(t)}{k_{ij}}(1-\delta_{il}) \binom{Y_l(t)}{k_{lj}} \bigg(\psi(Y_l(t)+1) - \psi(Y_l(t)-k_{lj}+1)\bigg).
	\]
\end{proposition}

\begin{proof}
	First, we define the $jl$-th element of the Jacobian matrix $\Lambda(\bm{Y}(t);\boldsymbol{\theta})$ as
	\[
	\frac{\partial}{\partial Y_l}\lambda_j(\bm{Y}(t);\boldsymbol{\theta}) = \frac{\partial}{\partial Y_l} \theta_j \prod_{i=1}^{p} \binom{Y_i(t)}{k_{ij}}.
	\]
	By applying the chain rule, we get
	\[
	\frac{\partial}{\partial Y_l} \prod_{i=1}^{p} \binom{Y_i(t)}{k_{ij}} = \bigg(\prod_{i=1}^{p} \binom{Y_i(t)}{k_{ij}}(1-\delta_{il})\bigg) \frac{\partial}{\partial Y_l} \binom{Y_l(t)}{k_{lj}},
	\]
	with $\delta_{il}$ an indicator function.
	Next, recalling the digamma function $\psi(x) = \dfrac{d}{dx} \log(\Gamma(x))$, we differentiate the logarithm of the binomial coefficient:
	\begin{align*}
		\frac{\partial}{\partial Y_l} \log \binom{Y_l(t)}{k_{lj}} &= \frac{\partial}{\partial Y_l} \left[ \log \Gamma(Y_l(t)+1) - \log \Gamma(k_{lj}+1) - \log \Gamma(Y_l(t)-k_{lj}+1) \right] \\
		&= \psi(Y_l(t)+1) - \psi(Y_l(t)-k_{lj}+1).
	\end{align*}
	Applying logarithmic differentiation, we have:
	\[
	\frac{\partial}{\partial Y_l} \binom{Y_l(t)}{k_{lj}} = \binom{Y_l(t)}{k_{lj}} \bigg(\psi(Y_l(t)+1) - \psi(Y_l(t)-k_{lj}+1)\bigg).
	\]
	Therefore, the $jl$-th element of the Jacobian matrix $\Lambda(\mathbf{Y}(t);\boldsymbol{\theta})$ is:
	\[
	\Lambda_{jl} = \theta_j \prod_{i=1}^{p} \binom{Y_i(t)}{k_{ij}}(1-\delta_{il}) \binom{Y_l(t)}{k_{lj}} \bigg(\psi(Y_l(t)+1) - \psi(Y_l(t)-k_{lj}+1)\bigg).
	\]
\end{proof}
%\clearpage
%\appendix
%\renewcommand{\thesection}{S.\arabic{section}} % Sections numbered as A.1, A.2, ...
%\begin{center}
	%\LARGE SUPPLEMENTARY MATERIAL
%\end{center}
%\addcontentsline{toc}{section}{SUPPLEMENTARY MATERIAL}

\section{LMA algorithm}\label{SM:algorithm}
Algorithm \ref{algorithm} gives the pseudo-code of the proposed nonlinear local mean-field (LMA) algorithm for parameter estimation in a quasi-reaction model.
\begin{algorithm}
	\caption{Nonlinear local mean-field (LMA) algorithm}
	\begin{algorithmic}
		\Data $\bm{Y}$, $K$, $V$ 
		\Initialization  $\boldsymbol{\hat{\theta}}_{0}$, $toll$, $maxIter$, $k = 0$, $H_0 =\mathbb{I}_r$
		\State Define $f(\boldsymbol{\theta}) = \sum_i^{T_c}\sum_c^{C}  ( \bm{Y}_{ci} - \bm{m}_{ci})^T ( \bm{Y}_{ci} - \bm{m}_{ci})$
		\While{$\|\nabla f(\boldsymbol{\theta}_k)\| > \epsilon$ \textbf{and} $k < maxIter$}
		\State $\boldsymbol{d}_k = -H_k \nabla f(\boldsymbol{\theta}_k)$ 
		\State $\alpha_k=\arg\min_\alpha f(\boldsymbol{\theta}_k + \alpha \boldsymbol{d}_k)$
		\State $\boldsymbol{\theta}_{k+1} = \boldsymbol{\theta}_k + \alpha_k \boldsymbol{d}_k$ \Comment update $\bm{m}$ according to \eqref{ODEsol}
		\State $\boldsymbol{s}_k = \boldsymbol{\theta}_{k+1} - \boldsymbol{\theta}_k$
		\State $\boldsymbol{g}_k = \nabla f(\boldsymbol{\theta}_{k+1}) - \nabla f(\boldsymbol{\theta}_k)$
	\State $H_{k+1} = \left(\mathbb{I}_r - \frac{\boldsymbol{s}_k \boldsymbol{g}_k^T}{\boldsymbol{g}_k^T \boldsymbol{s}_k}\right) H_k \left(\mathbb{I}_r - \frac{\boldsymbol{g}_k \boldsymbol{s}_k^T}{\boldsymbol{g}_k^T \boldsymbol{s}_k}\right) + \frac{\boldsymbol{s}_k \boldsymbol{s}_k^T}{\boldsymbol{g}_k^T \boldsymbol{s}_k}$
		\State $k = k + 1$
		\EndWhile
		\State $\boldsymbol{\hat{\theta}}_{LMA}= \boldsymbol{\theta}_k$
	\end{algorithmic}
	\label{algorithm}
\end{algorithm}

\section{Local Linear Approximation}\label{app:LLA}

In this section, we describe the Local Linear Approximation (LLA) approach for estimation of the rates \( \boldsymbol{\theta} \). The LLA estimates are used in the comparative study in section~\ref{sec:compLLA} and as initial values for Algorithm \ref{algorithm} that iteratively solves the optimization problem \eqref{prob}. 

The local linear approximation approach applies Euler's method to approximate the moments of the process at time $t$ conditional on the history of the process up to time $t$. Using the expression of the conditional moments derived from the chemical master equation \eqref{CME}, and assuming constant hazard rates between consecutive time points, the conditional moments are approximated as follows:
\begin{align}
	\mathbb{E}[\bm{Y}_{ci} \mid \bm{Y}_{c,i-1}] &\simeq \bm{Y}_{c,i-1} + \sum_{j=1}^{r} v_{lj} \lambda_j(\bm{Y}_{c,i-1}; \bm{\theta}) \Delta t_i \label{moments} \\
	\mathbb{E}[Y_{c,i,l} Y_{c,i,k} \mid \bm{Y}_{c,i-1}] &\simeq Y_{c,i-1,l} Y_{c,i-1,k} + \sum_{j=1}^{r} v_{lj} Y_{c,i-1,k} \lambda_j(\bm{Y}_{c,i-1}; \bm{\theta}) \Delta t_i \nonumber \\
	&\quad + \sum_{j=1}^{r} v_{kj} Y_{c,i-1,l} \lambda_j(\bm{Y}_{c,i-1}; \bm{\theta}) \Delta t_i + \sum_{j=1}^{r} v_{lj} v_{kj} \lambda_j(\bm{Y}_{c,i-1}; \bm{\theta}) \Delta t_i. \label{moments2}
\end{align}

Since the hazard function \( \lambda_j(\bm{Y}_{c,i-1}; \boldsymbol{\theta}) \) is linear in \( \bm{\theta} \), the regression model \eqref{generic} can be rewritten as:
\[
\Delta \bm{Y}_{ci} = M_{ci} \boldsymbol{\theta} + \boldsymbol{\varepsilon}_{ci}, \quad \boldsymbol{\varepsilon}_{ci} \sim \mathcal{N}(\mathbf{0}, \Omega_{ci}),
\]
where \( \Delta \bm{Y}_{ci} = \bm{Y}_{ci} - \bm{Y}_{c,i-1} \) is the vector of concentration changes between consecutive time points. The matrix \( M_{ci} \boldsymbol{\theta} = V \text{diag}(\boldsymbol{\lambda}_{ci}(\bm{Y}_{c,i-1}; \boldsymbol{\theta})) \Delta t_i \) and \( \Omega_{ci} = V \text{diag}(\boldsymbol{\lambda}_{ci}(\bm{Y}_{c,i-1}; \boldsymbol{\theta})) V^T \Delta t_i \) represent the discretized drift function and the dispersion matrix of the concentration differentiation process, respectively.

The local linear estimate \( \hat{\bm{\theta}}_{\text{LLA}} \) is then obtained by solving the following constrained generalized least-squares problem:
\[
\hat{\bm{\theta}}_{\text{LLA}} = \arg \min_{\bm{\theta} \geq \mathbf{0}_r} \sum_{i=1}^{T_c} \sum_{c=1}^{n} \left( \Delta \bm{Y}_{ci} - M_{ci} \bm{\theta} \right)^T \Omega_{ci}^{-1} \left( \Delta \bm{Y}_{ci} - M_{ci} \bm{\theta} \right).
\]

%In cases where the model involves several reactions sharing the same reactant (i.e., there are columns of \( V \) that are linear combinations of others), collinearity issues may arise. Such a scenario is discussed in Sections \ref{sec:xu} and \ref{sec:rhesus}. To address this, we introduce a small penalty parameter \( \phi = 1 \times 10^{-6} \) in \( \Omega_{ci} \).

\section{Standard error approximation}\label{SM:sd}
\begin{figure}[t!]
	\centering
	\includegraphics[scale=0.23]{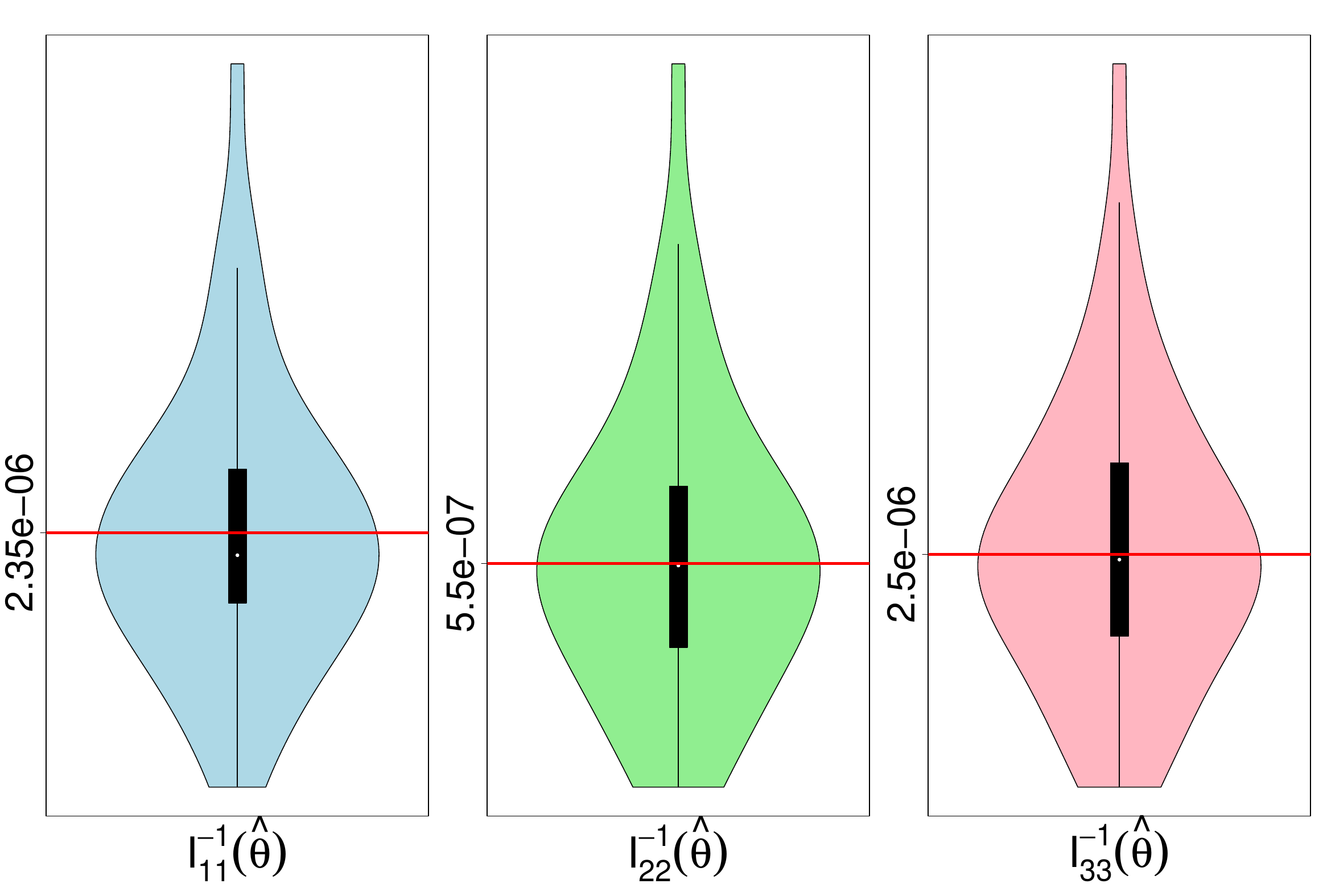}
	\centering
	\caption{Violin plots showing the distribution across 50 simulations of the diagonal elements of the inverse of the observed Fisher information matrix for the 3 reaction rates of the system in section~\ref{example}.  The horizontal line represents the empirical variance of the estimates. The figure shows a good agreement between the theoretically and empirically derived standard errors of the estimated parameters.}
	\label{fig:sd}
\end{figure}
In this section, we aim to validate the derivation of the standard error described in section~\ref{sec:sd}. As explained in that section, under a maximum-likelihood estimation framework, the variance-covariance matrix of $\hat{\boldsymbol{\theta}}$ is asymptotically approximated by the inverse of the observed Fisher information matrix,  for which we derive an explicit approximate formulation. We validate this derivation using the same model of section~\ref{example}. In particular, we generate data from this model with 100 replicates for each of the 50 simulated datasets. Figure \ref{fig:sd} presents violin plots corresponding to the diagonal elements of the inverse of the observed Fisher information matrix elements. The horizontal lines denote the observed variance of the estimates across the $50$ simulations. The figure shows a good agreement between the theoretically and empirically derived standard errors of the estimated parameters.

\section{Reaction systems used in the simulation study in section~\ref{sec:time}}\label{tabels}
This paragraph outlines the reaction systems used in section~\ref{sec:time} to evaluate the computational complexity of the algorithm.  Concerning the number of particles ($p$), for ease of notation the possible states are indicated only by the first three letters. When multiple of the three states are present, we denote these  with a subscript. 
\begin{figure}[h]
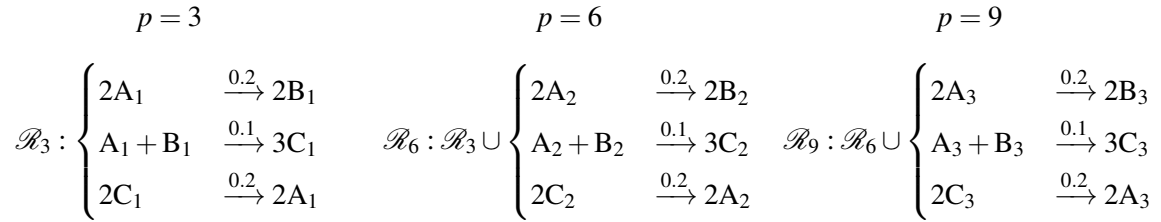

\begin{minipage}{.33\textwidth}
		\centering
		$p=3$
		\begin{align}
			\mathcal{R}_3  :
			\begin{cases}
				2 \text{A}_1 & \xrightarrow{0.2} 2\text{B}_1 \notag\\
				\text{A}_1 + \text{B}_1 & \xrightarrow{0.1} 3\text{C}_1 \notag\\
				2 \text{C}_1 & \xrightarrow{0.2} 2\text{A}_1 \notag
			\end{cases}
		\end{align}
	\end{minipage}
	\begin{minipage}{.33\textwidth}
		\centering
		$p=6$
		\begin{align}
			\mathcal{R}_{6}:	\mathcal{R}_3 \cup 
			\begin{cases}
				\label{CCRN}
				2 \text{A}_2 & \xrightarrow{0.2} 2\text{B}_2 \notag\\
				\text{A}_2 + \text{B}_2 & \xrightarrow{0.1} 3\text{C}_2 \notag\\
				2 \text{C}_2 & \xrightarrow{0.2} 2\text{A}_2 \notag
			\end{cases}
		\end{align}
	\end{minipage}	
	\begin{minipage}{.33\textwidth}
		\centering
		$p=9$
		\begin{align}
			&\mathcal{R}_{9}:	\mathcal{R}_{6} \cup 
			\begin{cases}
				2 \text{A}_3 & \xrightarrow{0.2} 2\text{B}_3 \notag\\
				\text{A}_3 + \text{B}_3 & \xrightarrow{0.1} 3\text{C}_3 \notag\\
				2 \text{C}_3 & \xrightarrow{0.2} 2\text{A}_3 \notag
			\end{cases}
		\end{align}
	\end{minipage}	
	\caption{Dynamic systems with an increasing number of particles ($p=3, 6, 9$).  Rates are reported above each reaction.}
	\label{tabp}
\end{figure} 
\begin{figure}[h]
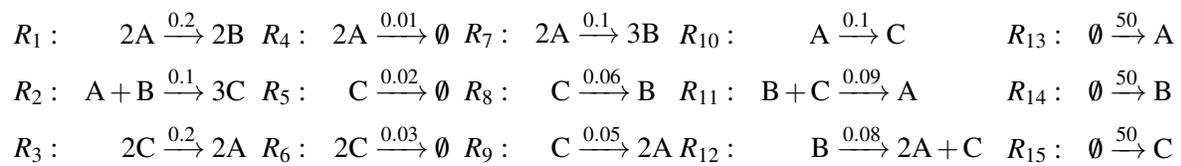

\begin{minipage}{.17\textwidth}
		\centering
		\begin{align}
			\begin{aligned}
				\begin{aligned}
					&R_1: & 2 \text{A} &\xrightarrow{0.2} 2 \text{B} \\
					&R_2: & \text{A} + \text{B} &\xrightarrow{0.1} 3 \text{C} \\
					&R_3: & 2 \text{C} &\xrightarrow{0.2} 2 \text{A}
				\end{aligned}
			\end{aligned}\notag
		\end{align}
	\end{minipage}	
\begin{minipage}{.17\textwidth}
		\centering
		\begin{align}
			\begin{aligned}
				\begin{aligned}
					&R_4: &2\text{A} &\xrightarrow{0.01}  \emptyset \\
					&R_5: &\text{C} &\xrightarrow{0.02}  \emptyset \\
					&R_6:& 2\text{C} &\xrightarrow{0.03}  \emptyset
				\end{aligned}
			\end{aligned}\notag
		\end{align}
	\end{minipage}	
	\begin{minipage}{.17\textwidth}
		\centering
		\begin{align}
			\begin{aligned}
				\begin{aligned}
					&R_7: & 2\text{A} &\xrightarrow{0.1} 3\text{B} \\
					&R_8: & \text{C}  &\xrightarrow{0.06}  \text{B} \\
					&R_9:& \text{C} &\xrightarrow{0.05} 2 \text{A}
				\end{aligned}
			\end{aligned}\notag
		\end{align}
	\end{minipage}	
	\begin{minipage}{.17\textwidth}
		\centering
		\begin{align}
			\begin{aligned}
				\begin{aligned}
					&R_{10}: &  \text{A} &\xrightarrow{0.1}  \text{C} \\
					&R_{11}: & \text{B} + \text{C} &\xrightarrow{0.09}  \text{A} \\
					&R_{12}:&  \text{B} &\xrightarrow{0.08} 2 \text{A}+\text{C}
				\end{aligned}
			\end{aligned}\notag
		\end{align}
	\end{minipage}	
	\begin{minipage}{.17\textwidth}
		\centering
		\begin{align}
			\begin{aligned}
				\begin{aligned}
					&R_{13}: & \emptyset &\xrightarrow{50}  \text{A} \\
					&R_{14}: &  \emptyset &\xrightarrow{50}  \text{B} \\
					&R_{15}:& \emptyset &\xrightarrow{50}  \text{C}
				\end{aligned}
			\end{aligned}\notag
		\end{align}
	\end{minipage}	
	\caption{Five dynamic systems are illustrated, with the first scenario including $\{R_1,R_2,R_3\}$ reactions. Each subsequent scenario adds 3 additional reactions to the previous one, culminating in a final system with 15 reactions. Reaction rates are reported above each reaction.}
	\label{tabr}
\end{figure}

\end{document}